\documentclass[letterpaper,11pt]{article}

\usepackage{bm}
\usepackage{url}
\usepackage{cite}
\usepackage{array}
\usepackage{color}
\usepackage{amsthm}
\usepackage{dsfont}
\usepackage{amsmath}
\usepackage{amssymb}
\usepackage{caption}
\usepackage{amsfonts}
\usepackage{booktabs}
\usepackage{colortbl}
\usepackage{graphicx}
\usepackage{mathrsfs}
\usepackage{multicol}
\usepackage{multirow}
\usepackage{setspace}
\usepackage{stfloats}
\usepackage{textcomp}
\usepackage{verbatim}
\usepackage{cellspace}
\usepackage{enumerate}
\usepackage{enumitem}
\usepackage{makecell}
\usepackage{bbm}
\usepackage{lineno}
\usepackage{authblk}
\usepackage{float} 
\usepackage{mathtools}
\usepackage{extarrows}

\usepackage{hyperref}
\hypersetup{
  linktoc      = all,
  colorlinks   = true, 
  urlcolor     = blue, 
  linkcolor    = blue, 
  citecolor    = red 
}
\usepackage[capitalize,noabbrev]{cleveref}
\usepackage[T1]{fontenc}
\usepackage[margin=1in]{geometry}
\usepackage[dvipsnames]{xcolor}
\usepackage[ruled,vlined,linesnumbered]{algorithm2e}
\usepackage[caption=false,font=normalsize,labelfont=sf,textfont=sf]{subfig}

\crefname{algocfline}{algorithm}{algorithms}
\Crefname{algocfline}{Algorithm}{Algorithms}

\setlist[itemize]{leftmargin=*}
\setlist[enumerate]{leftmargin=*}

\title{Instance-Optimality of Bidirectional PageRank Estimation\footnote{The work of both authors was supported by the VILLUM Foundation grant 54451, Basic Algorithms Research Copenhagen (BARC).}
}

\author[1]{Mikkel Thorup}
\author[2]{Hanzhi Wang\vspace{0.5em}}

\affil[1]{BARC, University of Copenhagen}

\affil[2]{The University of Melbourne\vspace{0.5em}}
\affil[1,2]{\{mikkel2thorup, hanzhi.hzwang\}@gmail.com}

\date{}

\newtheorem{theorem}{Theorem}[section]
\newtheorem{lemma}[theorem]{Lemma}

\def\cG{\mathcal{G}^*}

\def\eps{\varepsilon}
\def\tO{\widetilde{O}}

\def\Veps{V_{r}}

\def\Teps{T_{r}}

\def\Wout{W_{\mathrm{ext}}}
\def\Wisolated{W_{\mathrm{iso}}}

\def\nr{q}
\def\tO{\widetilde{O}}

\DeclareMathOperator{\polylog}{polylog}

\def\Nin{\mathcal{N}_{\mathrm{in}}}
\def\Nout{\mathcal{N}_{\mathrm{out}}}

\def\din{d_{\mathrm{in}}}
\def\dout{d_{\mathrm{out}}}

\def\Deltain{\Delta_{\mathrm{in}}}
\def\Deltaout{\Delta_{\mathrm{out}}}

\def\r{r}
\def\p{p}

\def\fmu{\mu}

\def\tpi{\tilde{\pi}}
\def\epi{\hat{\pi}}

\def\vpi{\pi}

\def\E{\mathrm{E}}
\def\Var{\mathrm{Var}}

\def\hset{V_h}

\def\rmax{r_{\mathrm{max}}}

\def\BiPR{\textup{\texttt{BiPR}}\xspace}

\def\push{\textup{\texttt{pushback}}\xspace}

\def\Ain{E_{\mathrm{in}}}

\def\indeg{\textup{\textsc{indeg}}\xspace}
\def\innbr{\textup{\textsc{in}}\xspace}
\def\noninnbr{\overline{\textup{\textsc{in}}}\xspace}
\def\outdeg{\textup{\textsc{outdeg}}\xspace}
\def\outnbr{\textup{\textsc{out}}\xspace}
\def\nonoutnbr{\overline{\textup{\textsc{out}}}\xspace}
\def\jump{\textup{\textsc{jump}}\xspace}

\def\adj{\textup{\textsc{adj}}\xspace}

\def\adapush{\textup{\texttt{AdaptiveBiPR}}\xspace}

\def\tp{\pi}
\def\Din{D_{\mathrm{in}}}

\def\pf{\delta}

\def\bpush{b_{\mathrm{push}}}
\def\vpush{v_{\mathrm{push}}}

\def\smallexpo{\gamma}

\begin{document}

\maketitle

\begin{abstract}
We study the problem of estimating a vertex's PageRank within a constant relative error, with constant probability. 
We prove that an adaptive variant of the simple classic bidirectional algorithm is instance-optimal up to a polylogarithmic factor for all directed graphs of order $n$ whose maximum in- and out-degrees are at most a constant fraction of $n$. 
In other words, there is no correct algorithm that can be faster than our algorithm on any such graph by more than a polylogarithmic factor. We further extend the instance-optimality to all graphs in which at most a polylogarithmic number of vertices have unbounded degrees. This covers all sparse graphs with $\tO(n)$ edges. 
In addition, we provide a counterexample showing that the bidirectional algorithm is not instance-optimal for graphs whose degrees are mostly equal to $n$. 
We also consider weighted graphs and multigraphs. We show that the bidirectional algorithm is instance-optimal on \emph{all} multigraphs, but for weighted simple graphs, we have almost the same limitations as for unweighted simple graphs. 

\end{abstract}

\section{Introduction} \label{sec:introduction}

In the last two decades, the problem of PageRank computation in its many forms has attracted considerable attention~\cite{chen2004local, fogaras2005towards, avrachenkov2007monte, gleich2007approximating, andersen2008local, bar2008local, bressan2011local, borgs2012sublinear, bressan2013power, borgs2014multiscale, lofgren2014fast, lofgren2016personalized, wang2018efficient, bressan2018sublinear, wang2020personalized, sigmod_LiaoLDW22, bressan2023sublinear, wang2024revisiting, LiuL24_hppr, bertram2025estimating, PODS_ssppr, kwok_yang_ITCS26, bertram2026undirected, soda_Thorup0W026}, and has been adopted to various applications, including webpage ranking~\cite{page1998pagerank}, recommender systems~\cite{wtf_GuptaGLSWZ13}, spam filtering~\cite{vldb_GyongyiGP04}, node classification~\cite{wang2021approximate}, link prediction~\cite{yin2019scalable}, and many others~\cite{gleich2015pagerank, berkhin2005survey}. 

PageRank is a classic graph centrality measure introduced in \cite{page1998pagerank}. 
The definition of PageRank builds on so-called {\em $\alpha$-discounted random walks}, where $\alpha \in (0,1)$ is a parameter. Such a random walk has a start vertex, predefined or randomly sampled. In each step, it terminates with probability $\alpha$, or continues to a uniformly chosen out-neighbor with probability $1-\alpha$. When a walk reaches a dangling vertex, it simply terminates there. To simplify definitions and algorithm descriptions, we assume that every dangling vertex has a self-loop so that all vertices have non-zero out-degree. This does not affect the termination of the walk, as it will still terminate at the dangling vertex. 
Following prior work~\cite{lofgren2014fast, bressan2023sublinear, wang2024revisiting, soda_Thorup0W026}, we assume that $\alpha$ is a constant. 

Given a directed graph $G=(V,E)$ and a vertex $t\in V$, the \emph{PageRank} of $t$, denoted $\pi(t)$, is defined as the probability that an $\alpha$-discounted walk starting from a uniformly random vertex terminates at $t$. A widely used variant is the Personalized PageRank (PPR): for vertices $v,t \in V$, the PPR of $t$ with respect to $v$, denoted $\vpi(v,t)$, is the probability that an $\alpha$-discounted walk starting from $v$ terminates at $t$. Let $n=|V|$ denote the number of vertices in $G$. Then by definition, 
\begin{align}\label{eqn:pagerank_ppr}
\vpi(t)=\sum_{v\in V}\vpi(v,t)/n. 
\end{align}

\paragraph{Problem formulation. } In this paper, we study the problem of estimating a vertex's PageRank. The input to our problem is $(G,t)$, where $G$ is an arbitrary directed graph and $t$ is an arbitrary vertex in $G$ called the target vertex. The goal is to estimate the PageRank of $t$ (i.e., $\vpi(t)$) within a constant relative error with constant probability. 

For our results, we focus on algorithmic computational complexity under the standard RAM model. To clarify what query operations for interacting with the graph are available, we consider the standard query model~\cite[Chapter 10]{books_Goldreich17} for sublinear graph algorithms (see~\Cref{sec:preliminaries} for details). We parameterize the complexity of the problem by $n$ and $m$, the number of vertices and edges in the graph $G$, respectively. The value of $n$ is part of the query model and is therefore available to the algorithms. We are not given any additional information about the graph.

\paragraph{A classic bidirectional search algorithm. } In this paper, we investigate the instance optimality~\cite{instance_optimal_Roughgarden19,book_Roughgarden_beyond} of an adaptive variant of a simple, classic bidirectional search algorithm for PageRank from 
\cite{lofgren2014fast,lofgren2016personalized}, here refer to as \BiPR. 

The \BiPR algorithm combines the $\push$ operation with forward Monte Carlo simulations, where $\push$ is a canonical technique proposed in~\cite{andersen2008local} to propagate random-walk probability mass backward from the target $t$ along in-edges. In comparison, the Monte Carlo simulations generate $\alpha$-discounted random walks forward from uniformly chosen source vertices along out-edges. 
This bidirectional search algorithm was formally introduced in 2014~\cite{lofgren2014fast}, originally designed to estimate $\vpi(s,t)$ for a given pair of vertices $s$ and $t$, and it can be trivially extended to estimate $t$'s PageRank $\pi(t)$ (simply add a new source $s$ with edges to all nodes in the original graph). 
A simplified version, was later proposed in 2016~\cite{lofgren2016personalized}, optimizing the way $\push$ and Monte Carlo simulations are combined. It is this simplified version, applied to PageRank that we refer to as \BiPR.

Due to the difficulty of analyzing the time cost of $\push$, for many years, only an average-case upper bound of $O(m^{1/2})$ ~\cite{lofgren2016personalized} was known for \BiPR. This is when averaging the complexity over all target vertices $t$ in a worst-case graph $G$. Considering the worst-case complexity of PageRank, for given $G$ and $t$, for a long time, the best bound in terms of $n$ and $m$ was the $O(n)$ time achieved just using Monte Carlo simulations.  However, for $m\ll n^2$,~\cite{bressan2018sublinear} improved the complexity to $\tO(n^{5/7}m^{1/7})$ by using a more complicated version of the bidirectional approach, later improving it to $\tO(n^{2/3}m^{1/6})$~\cite{bressan2023sublinear}.  More recently, \cite{wang2024revisiting} revisited the original $\BiPR$ algorithm and established a complexity of $O(n^{1/2}  m^{1/4} )$.
They proved that this complexity is worst-case optimal by constructing hard instance graphs, where any algorithm must take at least $\Omega(n^{1/2} m^{1/4})$ time on the hard graphs to ensure the derived estimate is within a constant relative error with constant probability. 

We shall return with more discussion of related work in Section \ref{sec:other-relate}, including bounds related to the maximal in- and out-degrees. 
For now we just conclude that the \BiPR algorithm has a worst-case optimal complexity of $O(n^{1/2}  m^{1/4} )$, and for sparser graphs, a much better average-case complexity of $O(m^{1/2})$.

\subsection{Instance-Optimality}\label{sec:instance-optimal}
As mentioned previously, we will investigate the instance-optimality of the \BiPR algorithm. Instance-optimality is the ultimate notion of optimality~\cite{colt_ChenL16_openproblem,instnce-optimal_AfshaniBC17}. 
It is an extremely appealing beyond-worst-case guarantee that an algorithm can have, yet only in rare cases does a problem admit an instance-optimal solution
\cite{instance_optimal_Roughgarden19, book_Roughgarden_beyond,sosa_bidirectional_HaeuplerHRTT25}. 

More specifically, instance-optimality is the dream of an algorithm that is best, or close to best, on every possible instance, including easy instances, e.g., instances solvable in logarithmic time. This contrasts with worst-case optimal algorithms which are best on worst-case instances, but are often beaten in practice on real instances.

Above, when we say that an algorithm is best on an instance, we are only competing with {\em correct algorithms} that output correct answers on every {\em valid input instance}. This includes an {\em instance-smart} algorithm that may be tuned to be fast for particular instances. However, to be correct, it must always \emph{verify} that the answer is correct for any given input instance (which may or may not be the one it was tuned for). 
Thus, an algorithm is instance optimal only if no correct (possibly instance-smart) algorithm is substantially faster on any valid instance. 

Worst-case optimal algorithms exist for all problems since by definition the worst-case complexity of a problem is the best worst-case complexity of an algorithm. However, for many problems it is easy to show that no instance-optimal algorithm exists, illustrated below.

\paragraph{Impossibility of instance-optimality and the power of instance-smart algorithms.} 
Consider a simple informative example where instance-optimality is impossible. The problem is that of deciding if an array $B$ with $n$ entries has a duplicate element.
Using randomization, we can solve the problem in $O(n)$
expected time using a hash table. To verify a no-instance with constant probability, we need to probe $\Omega(n)$ entries, so $\Theta(n)$ is the expected worst-case complexity of the problem.
However, if we had an instance-optimal algorithm, then we could solve all instances in constant time!

More specifically, we have a yes-instance if and only if we have indices $i$ and $j$ such that $B[i]=B[j]$. An instance-smart algorithm could have two ``lucky'' numbers $i$ and $j$. It would first check if $B[i]=B[j]$, reporting yes in this case; otherwise it could just check the whole array for duplicates. 
An instance-optimal algorithm $A$ has to be best, or close to best, on all instances, so it would have to answer all yes-instances within some constant, or close to constant, time $C$. We could then create an algorithm $A^+$ that on any instance ran $A$ for $C$ steps. If $A$ says yes, $A^+$ says yes, and if $A$ says no, or does not terminate within $C$ steps, then $A^+$ says no, so $A^+$ solves all instances within $O(C)$ time. This contradicts the $\Omega(n)$ lower bound and demonstrates the impossibility of instance optimality for this problem.

A similar argument says that an instance-optimal algorithm for any NP-hard decision problem would imply NP$=$P, and instance-optimality is actually much stronger, as it requires the algorithm to be efficient on any instance where the answer can be verified quickly. More explanations can be found in~\cite[Section 3.4.1]{books_Barbay20}. 

In this paper, we seek instance-optimality within a polylogarithmic factor. 
This slightly relaxes the requirement compared to instance optimality within a constant factor, but does not fundamentally affect the hardness of the problem, and we still encounter the same impossibility issues discussed above.
For example, for the duplicate detection problem, instance-optimality within a polylogarithmic factor would imply a polylogarithmic-time algorithm for all instances, contradicting that the worst-case complexity is linear.

\paragraph{Some  known positive examples of instance-optimality.}

While instance-optimality is not generally possible, some 
very nice cases of instance-optimal algorithms are known, e.g., in quantity estimation~\cite{siamcomp_DagumKLR00}, in database aggregation~\cite{FaginLN03_instance_optimal}, in sequential estimation~\cite{stoc_ValiantV16_instance_optimal, dang_neurips_23, focs_NarayananRTT24_thorup}, in computational geometry~\cite{instnce-optimal_AfshaniBC17}, 
in distribution testing and learning~\cite{siamcomp/ValiantV17, books_ValiantV20, focs_BlancCW25_clement}, in shortest paths~\cite{sosa_bidirectional_HaeuplerHRTT25}, and in contextual bandits~\cite{nips_LiRNJJ22}. There is also a line of work on instance-optimality in differential privacy~\cite{SIGMOD_best_DP, DongYi_DP_instanceopt, huang2021instance, AD_dp_instanceopt, ccs_dp_instanceopt}, where worst-case optimality is trivial to achieve but meaningless in practice, as typical instances are much easier than worst-case instances.

Most related to our work is the recent result \cite{sosa_bidirectional_HaeuplerHRTT25}   that the bidirectional Dijkstra algorithm is instance-optimal in the number of probed edges for deciding the shortest path from $s$ to $t$.
The computational complexity may exceed the edge-probe lower bound by a logarithmic factor due to the use of priority queue operations. 

Interestingly the above instance-optimality for bidirectional Dijkstra is only for weighted multi-graphs with positive real weights (no zeros). Having parallel weighted edges when computing shortest paths, may seem redundant, but recall that the biggest challenge for an instance-optimal algorithm is if life is too easy for an instance-smart algorithm.
Dijkstra's algorithm relaxes all incident edges when visiting a vertex. 
We would like to force the instance-smart algorithm to do more or less the same.
Without parallel edges, you would know the one-hop distance from $u$ to $v$ as soon as you have found an edge from $u$ to $v$ and read its weight. However, with parallel edges, you have to read the whole incidence out-list from $u$ (or in-list from $v$), to make sure there isn't a lighter edge (which is always possible since edge-weights are positive). 

For contrast, in the easy case where the valid inputs 
are restricted unweighted graphs, that is, all weights are $1$, the instance-optimality from \cite{sosa_bidirectional_HaeuplerHRTT25} is off by a factor $O(\Delta)$ where $\Delta$ is the maximal degree; for in this case, an instance-smart algorithm can verify if the distance from $s$ to $t$ is $1$, simply by pointing to the edge $(s,t)$.

Related to \cite{sosa_bidirectional_HaeuplerHRTT25}, we note that~\cite{FOCS_best_HaeuplerHRTT24} establishes universal optimality for a new priority-queue in Dijkstra's single-source shortest paths problem. This is a notion strictly weaker than instance-optimality but still substantially stronger than worst-case optimality.

\subsection{Our Contribution: Instance-Optimality in PageRank}\label{sec:contribution}

In this paper, we study a simple adaptive variant of the aforementioned $\BiPR$ algorithm from~\cite{lofgren2016personalized}. Our main contribution is to prove that this adaptive $\BiPR$ is instance-optimal within a polylogarithmic factor, not for all graphs, but for a large class of graphs including all sparse graphs. Our analysis shows that even though the running time of $\BiPR$ is highly dependent on the instance, varying between $O(\log n)$ and $O(n)$, its running time is always close to the best possible on every one of these graphs. Thus, contrasting worst-case complexity, where we just have to analyze worst-case instances, here we have to prove upper- and lower-bounds for all instances for which we claim instance-optimality. This instance-complexity view will be explained more clearly in~\Cref{subsec:technique_analysis}.

\paragraph{Instance-optimality and good (statistically correct) PageRank algorithms}
Recall that instance-optimal algorithms should only compete with correct algorithms. In the probabilistic setting of PageRank computation, the correctness is statistical, and for brevity, we shall call such algorithms ``good'', as defined below.

A \emph{valid} input instance to PageRank is a pair $(G,t)$ consisting of an arbitrary directed graph $G$ and an arbitrary target vertex $t$ in $G$. We say a PageRank algorithm $A$ is \emph{good} if for every valid input $(G,t)$, the algorithm estimates $\pi(t)$  within a constant factor with constant probability.

A good PageRank algorithm $A^*$ is \emph{instance-optimal} on a given instance $(G,t)$ if there is no good algorithm $A$  that on $(G,t)$ is faster than $A^*$  by more than a polylogarithmic factor. 
The dream would be to find a good algorithm $A^*$ that is instance-optimal on all valid instances. However, recall from the discussion in Section \ref{sec:instance-optimal} that such an instance-optimal algorithm often does not exist. 

Here we show that there are indeed
graphs on which the
$\BiPR$-algorithm is not
instance-optimal. However, instead of just giving up the dream of instance-optimality, we establish that $\BiPR$ is instance-optimal on a rich class of graphs $\cG$, which includes all sparse graphs and almost all other graphs, as detailed below.

\paragraph{The class $\cG$ including all sparse graphs and almost all other graphs.}

Our contribution is to show that the adaptive $\BiPR$ is instance optimal for a large class $\cG$ consisting of all graphs where at most a polylogarithmic number of vertices may have a very high in- or out-degree of $(1-o(1))n$. This includes all sparse graphs
because a graph with $n$ vertices and $m$ edges can have at most $4m/n$
vertices with in- or out-degree above $n/2$, and sparse graphs have $m=\tilde O(n)$. It also includes
almost all graphs, since a large random graph, w.h.p., has no vertices of degree above $0.6 n$. We note that $\cG$ also includes dense worst-case instances used in previous lower bounds from \cite{wang2024revisiting}; for inspection shows that they remain hard if we change $n$ to $n'=2n$,
adding $n$ extra vertices not connected to the original instance, and now the maximal degree is $n'/2$. For these worst-case instances, we have a lower-bound of $\Omega(n^{1/2} m^{1/4})$ which is 
$\Omega(n)$ for $m=\Omega(n^2)$.

\paragraph{Instance-optimality of $\BiPR$ on all graphs in $\cG$.}
The instance optimality of $\BiPR$ on $\cG$ means that it is always a near-optimal choice for all graphs in $\cG$. We note that the valid inputs are still all graphs, including graphs outside $\cG$. We only compare against algorithms that like $\BiPR$ are good on all graphs, when we say $\BiPR$ is near-optimal on all graphs in $\cG$.

The class $\cG$ includes not only hard instances for worst-case and average-case, but also easy instances on which we should perform much better. 
A very simple example is a graph where $t$ is only connected to itself by a loop. An instance-smart algorithm would recognize this case in constant time, and report that the PageRank is $1/n$. 
The $\BiPR$-algorithm is not designed for this easy case. It will
spend $O(\log n)$ time and report that
the PageRank is around $1/n$. This is, however, still within a polylogarithmic factor from best possible as we promised for every single instance in $\cG$.

\paragraph{Instance-bad case outside $\cG$.}
Complementing these positive results, we will show that $\BiPR$ is not instance optimal on all graphs. Specifically, we show that $\BiPR$  is not instance optimal on what we call \emph{mostly-degree-$n$} graphs where all but $o(n/\log n)$ vertices have in- and out-degree exactly $n$ (or $n-1$ without self-loops). In mostly-degree-$n$ graphs, all vertices have PageRank close to $1/n$.

An instance-smart algorithm can verify a mostly-degree-$n$ graph by sampling $O(\log n)$ vertices. If they 
all have in- and out-degree exactly $n$ (or $n-1$ without self-loops), then it can report that the PageRank is around $1/n$; otherwise it can just run a standard Monte Carlo simulation to estimate the PageRank.
This is indeed a good algorithm for all graphs, and it handles mostly-degree-$n$ graphs in $O(\log n)$ time.
However, $\BiPR$ will use $\tilde\Theta(n)$ time on mostly-degree-$n$ graphs. 

We can think of the issue as follows. For a mostly-degree-$n$ graph, the instance-smart algorithm will 
query the degree of a single vertex, and discover that it has edges to \emph{all} other vertices. This is a lot of edges to learn in constant time, and the instance-smart algorithm takes full advantage of this information. However, $\BiPR$ is not designed to exploit this kind of extreme information. It only learns about edges that it queries directly.

We could easily augment \BiPR to include the above test for mostly-degree-$n$ graphs, and perhaps other special cases, but our goal is not to construct a complicated algorithm that is instance-optimal for as many graphs as possible. Our main contribution is to show that (the adaptive variant of) the simple classic $\BiPR$-algorithm is instance-optimal for the class $\cG$. That being said, it is hard to imagine natural graphs outside $\cG$, except perhaps the complete graph $K_n$ (with bidirected edges). For that reason, in practice it might be worthwhile to include the above simple test before running \BiPR so that it becomes instance optimal on $\cG\cup\{K_n\}$.

\paragraph{Reflections on exceptions to instance-optimality.}
Recall from Section \ref{sec:instance-optimal}
that instance-optimality
is an extremely appealing beyond-worst-case guarantee that an algorithm can have, yet only in rare cases does a problem admit an instance-optimal solution
\cite{instance_optimal_Roughgarden19, book_Roughgarden_beyond,sosa_bidirectional_HaeuplerHRTT25}. 

For PageRank, we have the classic $\BiPR$-algorithm that was known to be worst-case optimal.
It is, however, clearly not instance-optimal in general, as illustrated above by the mostly-degree-$n$ graphs.
However, instead of just giving up the dream of instance-optimality, we prove that
we do have it if we restrict ourselves to $\cG$, not allowing too many vertices of very high degree.
\emph{To the best of our knowledge, this is the first natural example of an algorithm for a standard problem that is not 
instance-optimal for all instances, but instance-optimal for a rich and interesting subset of the  valid input instances}. 
Handling random instances may be easy, but $\cG$ includes all sparse graphs and they can have a lot of interesting structure.

A different approach is to change the problem to allow a wider set of valid input instances that is harder for an instance-smart algorithm. An example of this is the previously mentioned instance-optimality of bidirectional Dijkstra \cite{sosa_bidirectional_HaeuplerHRTT25}. Their instance-optimality is for the positively-weighted multigraphs.
This means that even if the input is a simple graph, a good/correct algorithm has to check for more
edges between the same pair of vertices, to be sure it has the lightest edge. In \cite{sosa_bidirectional_HaeuplerHRTT25}, they leave it as a possibly impossible open problem to get an instance-optimal algorithm for the natural case where only simple graphs are valid inputs.

\paragraph{Weighted graphs and multigraphs.}
We will also consider weighted graphs and multigraphs at the end of this paper. We will show that $\BiPR$ is instance-optimal on \emph{all} multigraphs, but for weighted simple graphs, we have almost the same limitations as for unweighted simple graphs (we only claim instance-optimality for weighted versions of graphs in $\cG$).

\subsection{Other Related Work}\label{sec:other-relate}
As discussed earlier,
when considering worst-case computational complexity parameterized solely by $n$ and $m$ under the standard adjacency-list query model on general directed graphs, \cite{wang2024revisiting}
proved that $\BiPR$ runs in $O(n^{1/2} m^{1/4})$ time, and that this is the best worst-case time for any algorithm. 
Improved worst-case bounds have been obtained only under additional assumptions, such as allowing more queries~\cite{bertram2025estimating}, exploiting additional knowledge of the graph (e.g., the maximum in-degree of $G$)\cite{wang2024revisiting,soda_Thorup0W026}, or restricting to undirected graphs\cite{lofgren2015bidirectional, wang2023estimating, wang_kdd2024_revisiting, bertram2026undirected, kwok_yang_ITCS26}. We shall discuss some of these results, and how they relate to our instance optimality.

\paragraph{Bounds on the maximal degrees.}
Let $\Deltain$ and 
$\Deltaout$ be the maximal in- and out-degree in the graph. Then \cite{wang2024revisiting} actually provides a more refined bound of  $\Theta(n^{1/2} \min\{\Deltain^{1/2}, \Deltaout^{1/2}, m^{1/4} \})$
for the running-time of $\BiPR$. 
Note that if $\Deltain,\Deltaout<n/2$, then the graph is in our class $\cG$ for which we know that $\BiPR$ is instance-optimal. 

However, our instance-optimality is only when we compare against algorithms that have to be good or statistically correct on all possible graphs, that is, any graph is a valid input. If we instead view  $\Deltain$ and $\Deltaout$ as constraints in the sense that only graphs with these degree bounds are valid inputs, then one might be able to do much better. 

When $\Deltain$ and $\Deltaout$ are constraints on the degrees in the valid input, then a
very recent work~\cite{soda_Thorup0W026} presents a novel randomized backward propagation technique which only propagates selectively based on Monte Carlo estimated PageRank scores. The new method improves the worst-case computational complexity to $\tilde{\Theta}\left(n^{1/2}\min\left\{ \Deltain^{1/2} \big/ n^{\smallexpo},\Deltaout^{1/2} \big/ n^{\smallexpo},m^{1/4}\right\}\right)$, where $\smallexpo = \frac{1}{2} \left(2\max\left\{\log_{1/(1-\alpha)}\Deltain,1\right\}-1\right)^{-1}$. When $\Deltain=O(1)$, this result is polynomially better than the $\Theta\left(n^{1/2}\min\left\{ \Deltain^{1/2},\ \Deltaout^{1/2},\ m^{1/4} \right\}\right)$ established in~\cite{wang2024revisiting} for $\BiPR$.

The new randomized backward propagation technique from~\cite{soda_Thorup0W026} also
applies in the unconstrained case with $\Deltain=\Deltaout=n$, and one might wonder if it could be better than the $\BiPR$ algorithm, at least for some graphs. However, our instance-optimality says that it cannot be substantially better for any graph in $\cG$.

\paragraph{Undirected graphs.}
There is also a line of work~\cite{lofgren2015bidirectional, wang2023estimating, wang_kdd2024_revisiting, bertram2026undirected, kwok_yang_ITCS26} that considers PageRank computation on undirected graphs. It is shown in \cite{wang_kdd2024_revisiting, bertram2026undirected, kwok_yang_ITCS26} that $\Theta(m^{1/2})$ characterizes the worst-case complexity of the problem on undirected graphs under the standard adjacency-list model. 
While $\BiPR$ can be applied to undirected graphs, it cannot compete with algorithms that are allowed to exploit their guaranteed symmetry.

\paragraph{Other graph queries.}
There has also been work considering different graph queries. Our instance-optimality of $\BiPR$ holds even if we allow the adjacency query (is there an edge between $u$ and $v$?), but 
\cite{wang2020personalized} proposes sorting in-neighbors by their out-degree to accelerate the problem of estimating $\pi(s,t)$ for all $s\in V$ given a target vertex $t$. Using both adjacency queries and sorted in-neighbors improves the previously mentioned average-case upper bound from $\tilde{\Theta}(m^{1/2})$~\cite{lofgren2016personalized} to $\tilde{\Theta}(\min\{m^{1/2}, n^{2/3}\})$~\cite{bertram2025estimating}.

It is important for our instance-optimality that valid instances include all possible orderings of the in- and out-neighbors of each vertex. This means
that we do not assume
any specific ordering such as the above ordering by in-neighbors by out-degree. The point is that if we from a vertex $v$ were given the first neighbor $w$ in such an ordering, then we would know that $v$ did not have any neighbors prior to $w$ in this ordering, thus ruling out a lot of potential edges with a single query.

\subsection{Techniques: Instance Complexity Analysis}\label{subsec:technique_analysis}

To present our results more specifically, we fix a directed graph $G = (V, E)$ with target vertex $t$ whose PageRank we aim to estimate. For any $r\in [0,1]$, we define
\begin{align}\label{eqn:def_V_T}
\Veps=\{v\in V \mid \vpi(v,t) \ge r\}, \text{ and } \Teps=\sum_{v\in \Veps}\left(1+\din(v)\right), 
\end{align}
where $\din(v)$ denotes the in-degree of vertex $v$ in $G$. 
We also define: 
\begin{align}\label{eqn:T*}
T^*=\max_{r\in [0, 1]}\left\{\min\left\{\Teps, r/\vpi(t)\right\}\right\}. 
\end{align} 

We always have $T^*=O(n)$ since $\pi(t)\geq \alpha/n$. We claim that $T^*$ characterizes the \emph{instance complexity} of estimating $\pi(t)$ for all graphs in $\cG$, that is, $\tilde\Theta(T^*)$ is the best run-time an algorithm that is good on all graphs can have for the instance $(G,t)$ if $G\in\cG$. We shall come back with more intuition for $T^*$ in Section \ref{sec:intuition}.

\paragraph{Upper and lower bounds of instance complexity. }
On the upper-bound side, we prove that the adaptive variant of $\BiPR$ estimates $\pi(t)$ within a multiplicative factor of $(1 \pm 1/\log^{1/4} n)$ in expected time $O(T^* \log n)$, with probability at least $1 - 1/\log^{1/4} n$.

On the lower-bound side, 
suppose $G$ has at most $h$ vertices with in- or out-degree above $(1-\eps)n$.  We note that $h$ and $1/\eps$ are polylogarithmic for every graph in the class $\cG$ (see~\Cref{sec:contribution} for definition of the class $\cG$). 
We will prove a general lower-bound of  
\begin{equation}\label{eq:lower-intro}
\Omega((T^*(\eps/ (h+1)^2)/\log^{3/2} n)=_{\cG}\tilde\Omega(T^*)\textnormal,
\end{equation}
on the time it takes for any good algorithm to estimate $\pi(t)$. Note that multiple valid settings of $\varepsilon$ and $h$ may exist for a given graph, and the setting that yields the largest complexity bound constitutes our lower bound. 

Above, the $=_{\cG}\tilde\Omega(T^*)$ denotes that this last lower-bound is only for graphs in $\cG$. 
Combining this with our upper bound, we conclude that $\tilde{\Theta}(T^*)$ captures the instance-complexity of graphs in $\cG$ and that our adaptive $\BiPR$ is instance optimal for every graph in $\cG$.
We are going to explain more on the idea for the upper and lower bounds in Section \ref{sec:intuition}.

\paragraph{Instance-non-optimality for mostly-degree-$n$ graphs. }

As mentioned in Section~\ref{sec:contribution},
we will also demonstrate that $\BiPR$ is not instance optimal on all graphs. The bad example was the very dense graphs called \emph{mostly degree-$n$ graphs} in which all but $o(n/\log n)$ of its vertices have both in-degrees and out-degrees equal to $n$. This implies that all vertices have in- and out-degree close to $n$.
In such graphs, all vertices end up with PageRank close to $1/n$. An instance-smart algorithm can sample $O(\log n)$ vertices, and return the estimate $1/n$ if all their in- and out-degrees are $n$.  Otherwise, it can just revert to $\BiPR$.

For mostly degree-$n$ graphs we claim that $T^*=\Theta(n)$. To see this, consider $r=\alpha$. Then $\pi(t,t)\geq r$ since a walk from $t$ terminates instantly with probability $\alpha$. Therefore $t\in V_r$, and since all vertices have in-degree close to $n$, we conclude that $T_r=\Omega(n)$. We also have $r/\pi(t)\sim n$, so we conclude that $T^*=\Theta(n)$.

We will also show that $\BiPR$ takes $\Omega(T^*)=\Omega(n)$ time on these mostly degree-$n$ graphs, meaning that it is exponentially worse than the instance-smart algorithm. The fundamental
issue is that the $\BiPR$-approach does not take advantage of learning 
$2n$ edges, when discovering that a vertex has in- and out-degree $n$.

\subsection{Intuition for $T^*$ in the Upper- and Lower-Bounds}\label{sec:intuition} 

We will now discuss the intuition behind $T^*$ in relation to our upper- and lower-bounds.
Recall that 
\begin{align*}
T^*=\max_{r\in [0, 1]}\left\{\min\left\{\Teps, r/\vpi(t)\right\}\right\}, \quad \text{where} \quad  \Veps=\{v\in V \mid \vpi(v,t) \ge r\} \quad \text{and} \quad   \Teps=\sum_{v\in \Veps}\left(1+\din(v)\right). 
\end{align*}

To understand $\tilde{\Theta}(T^*)$ intuitively, let us consider a lower-bound idea for the time complexity of estimating $\pi(t)$. Suppose that there is a vertex $v$ with $\pi(v,t)\ge r$ and an unvisited in-neighbor $u$ of $v$ with out-degree 1. Now suppose that there is a set $W$ of $8n\pi(t)/r$ vertices that are also unvisited by the algorithm. Since these vertices have not been visited, we may modify their incident edges without changing anything the algorithm has observed.  In particular, imagine modifying the graph so that every vertex in $W$ has out-degree 1 and points to $u$. In the resulting graph $G^+$, these vertices contribute additional PageRank mass to $t$ through the path $W\to u\to v\leadsto t$, giving
\begin{align*}
\pi_{G^+}(t) \ge \sum_{w\in W}\frac{\pi_{G^+}(w,t)}{n}=\frac{(1-\alpha)^2|W|\pi_{G^+}(v,t)}{n}= \frac{8(1-\alpha)^2\pi_G(t)}{r}\,\pi_{G^+}(v,t) \geq
\frac{\pi_{G^+}(v,t)}{r}2 \,\pi_G(t).
\end{align*}
Above, the $(1-\alpha)^2$ factor accounts for the probability of not terminating during either of the two steps: moving from $W$ to $u$, and then from $u$ to $v$. The last inequality assumes $\alpha\leq 1/2$. Recall that we picked $v$ because $\pi_G(v,t)\geq r$. 
When constructing $G^+$, we will also ensure that $\pi_{G^+}(v,t)\geq r$ (which may require considering multiple options for $u$ and $v$). Then we conclude that $\pi_{G^+}(t)\geq 2\pi_G(t)$. 

Since the above modification involves only unvisited parts of the graph, the algorithm cannot distinguish $G$ from $G^+$. It therefore produces the same output on both graphs, contradicting the correctness of the algorithm. 
Therefore, to ensure correctness, we need to either visit all in-neighbors of all vertices in the set $V_r$ (thus costing $\Omega(\Teps)$ time), or make random jumps to $8n\pi(t)/r$ vertices to rule out the set $W$ of unvisited vertices (thus costing $n/(8n\pi(t)/r)=\Omega(r/\pi(t))$ time).

To understand the intuition behind our $\tilde{O}(T^*)$ upper bound, we adopt a bidirectional search algorithm that propagates random-walk probability mass backward from $t$ along in-edges, combined with forward Monte Carlo simulation implemented by jumping to uniformly random vertices. Intuitively, the time required to propagate probability mass to all in-neighbors of all vertices $v$ with $\pi(v,t)\geq r$ is $\tilde{O}(\Teps)$, while the time required to jump to a vertex in the set $W$ of size $8n\pi(t)/r$ is $\Theta((r/\pi(t)))$. Balancing these two costs yields the $\tilde{O}(T^*)$ running time. 

These are only a sketch of the idea. As we shall see, the real argument is much more complicated.

\section{Preliminaries} \label{sec:preliminaries}

\paragraph{Query model.}
In this paper, we work within the standard {\em adjacency-list} query model for sublinear graph algorithms~\cite[Chapter 10]{books_Goldreich17}. The model assumes that the graph is stored in incidence-list format; i.e., each vertex maintains a list of its outgoing edges and a list of its incoming edges, and we have query access to the vertices. The vertices are numbered from 1 to $n$ and are identified by their indices. 
The following queries are supported in the model, each with unit cost: 
\begin{itemize}
    \item Degree queries $\indeg(i)$ and $\outdeg(i)$: Given a vertex $i \in [1, n]$, the $\indeg(i)$ query returns the in-degree of the vertex $i$, and the $\outdeg(i)$ query returns the out-degree of the vertex $i$. 
    \item Neighbor queries $\innbr(i, j)$ and $\outnbr(i, k)$: Given a vertex $i \in [1, n]$ and an index $j\in [1, \din(i)]$, the $\innbr(i, j)$ query returns the $j$-th in-neighbor of vertex $i$; Analogously, given a vertex $i \in [1, n]$ and a index $k\in [1, \dout(i)]$, the  $\outnbr(i, k)$, returns the $k$-th out-neighbor of $i$. 
\end{itemize}

We remark that this query model is generally consistent with the {\em arc-centric graph-access model} used in prior work on PageRank computation~\cite{bressan2023sublinear, wang2024revisiting, soda_Thorup0W026}. In the more realistic arc-centric graph-access model, the vertex set is not just the indices $1,\ldots,n$. To access vertices, we either have to find them following neighbor links from $t$, or we use a special $\jump()$ query that returns a vertex chosen uniformly at random from the graph. In the adjacency-list model, the $\jump()$ query can be simulated by simply generating a uniformly random index $k \in [1,n]$ and returning the $k$-th vertex. 

The $\BiPR$-algorithm and our upper-bound are formulated in the slightly weaker but more realistic arc-centric graph-access model using the $\jump()$-query while our lower-bounds are in the cleaner adjacency-list model.

For our lower-bounds, it is important that valid inputs include all possible orderings of the incidence-lists. In particular, we do not assume orderings like  in \cite{wang2020personalized} where in-neighbors are sorted according to their out-degrees, or in \cite{TetekT22} where the neighbors are sorted according to a global hash-order. However, our lower-bounds do not require general permutation-obliviousness as it is sometimes needed in connection with instance-optimality~\cite{instnce-optimal_AfshaniBC17, focs_NarayananRTT24_thorup}. 

\paragraph{Additional queries allowed in the lower-bounds.}
When establishing our lower bounds, we even allow for some more types of queries used in prior work. One is the adjacency query $\adj(u, v)$, which returns $1$ if there is an edge between vertices $u$ and $v$, and $0$ otherwise (i.e., $\adj(u, v)=1$ if $(u,v)\in E$ or $(v,u)\in E$). Another involves access to non-adjacency lists, including $\noninnbr(v, i)$ and $\nonoutnbr(v, i)$, which return the $i$-th vertex in the graph that is not an in-neighbor or out-neighbor of $v$, respectively. 
These queries are efficient for a vertex with degree above $n/2$, where we can specify only the non-neighbors.
Our lower bounds hold when we include these queries while our upper bound does not need them. 
We do not consider these additional queries when extending our results to weighted graphs and multigraphs.

\paragraph{Notations. } We summarize some frequently used notations below. We use $n = |V|$ and $m = |E|$ to denote the number of vertices and edges in the underlying graph, respectively. 
Throughout the paper, we assume that the value of $n$ is sufficiently large and known in advance. 
Additionally, we use $\din(v)$ and $\dout(v)$ to denote the in-degree and out-degree of a vertex $v$, respectively. We also use $\Nin(v)$ and $\Nout(v)$ to denote the in- and out-adjacency lists of $v$, respectively. Moreover, for any subset of vertices $U \subseteq V$, we define $\Nout(U)$ as the set of all out-neighbors of vertices in $U$, i.e., $\Nout(U)=\bigcup_{u\in U}\Nout(u)$. 

Additionally, for any subset $X$, where $X$ may be a set of vertices ($X \subseteq V$) or edges ($X \subseteq E$), we define $\vpi(u,X,v)$ to be the probability that an $\alpha$-discounted random walk starting from $u$ visits at least one element of $X$ and eventually terminates at $v$. Conversely, $\vpi(u,\overline{X},v)$ denotes the probability that the walk terminates at $v$ without visiting any element of $X$. Thus, for all $u,v$ we have
\begin{align}\label{eqn:relation_W_barW}
\vpi(u,v)=\vpi(u,X,v)+\vpi(u,\overline{X},v).
\end{align}

We are going to consider several modifications $G'$ to our graph $G$.
In our notations, when it is not clear which graph we are talking about, we will specify it as a subscript, e.g., we use $\vpi_{G'}(t)$ to denote the PageRank of vertex $t$ in $G'$. For any algorithm $A_R$ and random seed $R$, we use $\epi_{A_R(G')}(t)$ to denote the estimate for $\vpi(t)$ produced by running the algorithm $A$ with random seed $R$ on graph $G'$. 
Finally, throughout the paper, $\log$ denotes the logarithm to base $2$ unless otherwise specified.

\paragraph{Paper organization. } In the remainder of the paper, we first establish our lower bound in~\Cref{sec:instance_lowerbound_with_n_constraints} without including $h$ for simplicity, and then present our upper-bound algorithm in~\Cref{sec:upperbound}. In~\Cref{sec:counterexample}, we provide a counterexample showing that our algorithm is not instance-optimal for mostly-degree-$n$ graphs. In~\Cref{sec:lowerbound_sparsegraph}, we complement our lower bound results by incorporating $h$. Finally, in~\cref{sec:extension_multigraphs_weighted}, we extend our results to multigraphs and weighted simple graphs.
\section{Lower Bounds} \label{sec:instance_lowerbound_with_n_constraints}

This section presents our main result on the instance-optimality lower bound. For simplicity, we assume for now that $h=0$.

\begin{theorem}\label{thm:lowerbound_withn}
Consider any directed graph $G$ with maximal in- and out-degrees upper bounded by $(1-\eps)n$ for some $\eps\in[0,1]$. For any $r\in [0,1]$, suppose there exists an algorithm $A$ that estimates $\pi(t)$ in expected time $O\left(\min\left\{\Teps/\log^{1/2} n,\  r/\vpi(t)\cdot \eps /\log^{3/2} n\right\}\right)$. We can then construct a graph $G^+$, such that
\begin{align*}
\pi_{G^+}(t)=\omega(\pi_{G}(t)), \quad \text{and}\quad \Pr_R\left\{\epi_{A_R(G^+)}(t)=\epi_{A_R(G)}(t)\right\}\ge 1-o(1), 
\end{align*}
where the probability $\Pr_R$ is taken over the choice of the random seed $R$ used by the algorithm $A$. 
\end{theorem}

\Cref{thm:lowerbound_withn} implies that any good algorithm, which estimates $\vpi(t)$ within a constant factor with constant probability, requires an expected running time of 
\begin{align}\label{eqn:contradiction}
\Omega\left(\max_{r\in [0,1]}\left\{\min\left\{\Teps/\log^{1/2} n,\  r/\vpi(t)\cdot \eps /\log^{3/2} n\right\} \right\}\right).   
\end{align}
This establishes a lower bound of $\Omega\left(T^* \eps/\log^{3/2} n\right)$, 
where $T^*=\max_{r\in [0,1]}\left(\min\left(\Teps, r/\vpi(t)\right)\right)$ as defined in~\Cref{eqn:T*}.

\subsection{Some technical lemmas} \label{subsec:technical_lemma}
Before delving into the details of our lower-bound proof, we first present some technical lemmas that show certain graph modifications do not significantly change the $\alpha$-discounted random walk probabilities. Notably, these lemmas are general: they hold for any directed graph without degree constraints, and are independent of our lower-bound graph construction. 
To the best of our knowledge, these lemmas are established here for the first time. In a first reading, the reader may want to skip the proofs, and move on to the lower-bound construction. 

\begin{lemma}\label{lem:edge_subdividing}
For any vertices $s,t$ in any directed graph, subdividing any edge $(u,v)$ in the graph (i.e., introducing a vertex $v'$ and replacing the edge $(u,v)$ with the edges $(u,v')$ and $(v',v)$) can only decrease $\pi(s,t)$ by a factor at most $(2-\alpha)/(1-\alpha)$. This also allows all cases where some of the vertices $s, t, u, v$ are identical. 
\end{lemma}

\begin{proof}
Let $H$ and $H'$ denote the input graph and the graph after subdividing edge $(u,v)$, respectively. We then use $\vpi_{H}(s,t)$ and $\vpi_{H'}(s,t)$ to denote the values of $\vpi(s, t)$ on graph $H$ and $H'$, respectively. We will prove that for any vertices $s$ and $t$ in $H$: 
\begin{align}\label{eqn:vt}
\vpi_{H'}(s, t)\ge \vpi_{H}(s, t)(1-\alpha)/(2-\alpha). 
\end{align}

We begin by considering the simple and special case where the graph $H$ only consists of a single vertex with a self-loop, i.e., $s=u=v=t$. In this case, we have
\begin{align*}
\vpi_H(s, t)=\sum_{i=0}^{\infty} \alpha (1-\alpha)^i=1, \quad \text{ and }\quad \vpi_{H'}(s, t)=\sum_{i=0}^{\infty} \alpha (1-\alpha)^{2i}=1/(2-\alpha). 
\end{align*}
This is consistent with \Cref{eqn:vt}. 

We now extend the proof to cover all cases. 

We define the edge sets $X=\{(u,v)\}$ and $X'=\{(u,v'), (v', v)\}$. By~\Cref{eqn:relation_W_barW}, we have:
\begin{align*}
\vpi_H(s, t)=\vpi_H(s, \overline{X}, t)+\vpi_H(s, X, t), \quad \text{ and } \quad \vpi_{H'}(s, t)=\vpi_{H'}(s, \overline{X'}, t)+\vpi_{H'}(s, X', t). 
\end{align*}
Since the graph $H$ differs from graph $H'$ only by subdividing the edge $(u,v)$, we have $\vpi_H(s, \overline{X}, t)=\vpi_{H'}(s, \overline{X'}, t)$. We therefore focus on comparing $\vpi_H(s, X, t)$ and $\vpi_{H'}(s, X', t)$ in what follows. If $\vpi_H(s, X, t)\le \vpi_{H'}(s, X', t)$, then $\vpi_H(s,t)\le \vpi_{H'}(s,t)$ and \Cref{eqn:vt} holds trivially. Thus, it suffices to consider the case $\vpi_H(s, X, t)>\vpi_{H'}(s, X', t)$. In this case, 
\begin{align}\label{eqn:bridge_X}
\frac{\vpi_H(s, t)}{\vpi_{H'}(s, t)}=\frac{\vpi_H(s, \overline{X}, t)+\vpi_H(s, X, t)}{\vpi_{H'}(s, \overline{X'}, t)+\vpi_{H'}(s, X', t)}\le \frac{\vpi_H(s,X,t)}{\vpi_{H'}(s, X', t)}.  
\end{align}
We are going to show that $\frac{\vpi_H(s,X,t)}{\vpi_{H'}(s, X', t)}\le (2-\alpha)/(1-\alpha)$ in the following. 

For any vertex $x$, let $P_H(x)$ denote the probability that an $\alpha$-discounted random walk in the graph $H$, starting from $x$, reaches $u$ without ever traversing the edge $(u,v)$, and that, conditional on not terminating at $u$, its next step traverses $(u,v)$. Thus, $P_H(s)(1-\alpha)$ is exactly the probability that a walk starting from $s$ reaches $u$ without using $(u,v)$ and then traverses $(u,v)$ in the next step. Moreover, $(P_H(v)(1-\alpha))^{i-1}$ represents the probability that a walk starting from $v$ repeatedly moves to $u$ without using $(u,v)$ and then returns to $v$ by traversing $(u,v)$, repeating this cycle $i-1$ times. Consequently, in the graph $H$, the probability that a walk from $s$ to $v$ passes through $(u,v)$ $i>0$ times is
\begin{align}\label{eqn:i_times_H}
P_H(s)(1-\alpha)\big(P_H(v)(1-\alpha)\big)^{i-1}.
\end{align}

Recall that the edge set $X={(u,v)}$, and that $\vpi_H(s,X,t)$ and $\vpi_H(s,\overline{X},t)$ denote the probabilities that an $\alpha$-discounted random walk from $s$ terminates at $t$ with and without visiting any edge in $X$ (here in this case it is just the edge $(u,v)$), respectively. We therefore have
\begin{align}\label{eqn:compare1}
\vpi_H(s,X,t)
= \left(\sum_{i=1}^{\infty}
P_H(s)(1-\alpha)\big(P_H(v)(1-\alpha)\big)^{i-1}\right)
\vpi_H(v,\overline{X},t)
= \frac{P_H(s)(1-\alpha)\vpi_H(v,\overline{X},t)}{1-P_H(v)(1-\alpha)} .
\end{align}

Analogously, for any vertex $x$ in the graph $H'$, let $P_{H'}(x)$ denote the probability that an $\alpha$-discounted random walk in $H'$, moving from $x$ to $u$ without ever traversing the edges $(u,v’)$ or $(v’,v)$, and that, conditional on not terminating in the next two steps, it subsequently traverses the path $u \to v' \to v$. Hence, in $H’$, the probability that a walk from $s$ to $v$ passes through the path $u \to v' \to v$ $i>0$ times is
\begin{align*}
P_{H'}(s)(1-\alpha)^2\big(P_{H'}(v)(1-\alpha)^2\big)^{i-1}.
\end{align*}
Compared with \Cref{eqn:i_times_H}, here the factor $(1-\alpha)^2$ appears because $\dout(v')=1$ and traversing $u \to v' \to v$ requires two steps.

Recall that the edge set $X'=\{(u,v'), (v', v)\}$, and that $\vpi_{H'}(s,X',t)$ and $\vpi_{H'}(s,\overline{X'},t)$ denote the probabilities that an $\alpha$-discounted random walk from $s$ terminates at $t$ with and without visiting any edge in $X'$, respectively. It follows: 
\begin{align}\label{eqn:compare2}
\vpi_{H'}(s, X', t)= \left(\sum_{i=1}^{\infty}P_{H'}(s)(1-\alpha)^2(P_{H'}(v) (1-\alpha)^2)^{i-1}\right)\vpi_{H'}(v, \overline{X'}, t)=\frac{P_{H'}(s)(1-\alpha)^2\vpi_{H'}(v, \overline{X'}, t)}{1-P_{H'}(v)(1-\alpha)^2}. 
\end{align}

We now compare \Cref{eqn:compare1} and \Cref{eqn:compare2}. Observe that $P_H(x)=P_{H’}(x)$ for every vertex $x$, since these probabilities concern random walks that never traverse any of the edges $(u,v)$, $(u,v')$, or $(v',v)$. For simplicity, we therefore drop the subscripts and write $P(x)$. Then we have
\begin{align*}
\frac{\vpi_H(s,X,t)}{\vpi_{H'}(s, X', t)}
&= \frac{1-P(v)(1-\alpha)^2}{(1-\alpha)(1-P(v)(1-\alpha))}=\frac{1-P(v)(1-\alpha)+P(v)(\alpha-\alpha^2)}{(1-\alpha)(1-P(v)(1-\alpha))}\\
&=\frac{1}{(1-\alpha)}+\frac{P(v)(\alpha-\alpha^2)}{(1-\alpha)(1-P(v)(1-\alpha))}. 
\end{align*}
This is maximized when $P(v)$ achieves $1$ since $P(v)\in [0,1]$. This implies: 
\begin{align*}
\frac{\vpi_H(s,X,t)}{\vpi_{H'}(s, X', t)}\le \frac{1}{(1-\alpha)}+\frac{(\alpha-\alpha^2)}{(1-\alpha)(1-(1-\alpha))}=\frac{2-\alpha}{1-\alpha}. 
\end{align*}
Combining with~\Cref{eqn:bridge_X} completes the proof. 
\end{proof}

\begin{lemma}\label{lem:either}
For any subset of vertices $U$ and any vertices $s, t$ in a directed graph with $t \notin U$, 
\begin{align*}
\max\{\pi(s, \overline{U}, t),\;\pi(\fmu(U), \overline{U}, t)\} = \Omega(\pi(s, t)),
\end{align*}
where $\fmu(U) = \arg\max_{v \in \Nout(U)\setminus U} \pi(v, \overline{U}, t)$. This result also holds when $s=t$. 
\end{lemma}

\begin{proof}
Recall from~\Cref{eqn:relation_W_barW} that
$\vpi(s,t)=\vpi(s,U,t)+\vpi(s,\overline{U},t)$. 
If $\vpi(s,U,t)< \vpi(s,t)/2$, then we already have $\vpi(s,\overline{U},t)\ge \vpi(s,t)/2=\Omega(\vpi(s,t))$. If instead  $\vpi(s,U,t)\ge \vpi(s,t)/2$, we will show:  
\begin{align}\label{eqn:bound_U}
\vpi(\fmu(U),\overline{U},t)\ge \alpha \vpi(s,U,t). 
\end{align}
Substituting $\vpi(s,U,t)\ge \vpi(s,t)/2$ into this inequality gives
$\vpi(\fmu(U),\overline{U},t) \ge \alpha \pi(s,t)/2=\Omega(\pi(s,t))$. 
Combining these two cases completes the proof of~\Cref{lem:either}. 

To prove~\Cref{eqn:bound_U}, we recall that $\vpi(s, U, t)$ denotes the probability that an $\alpha$-discounted random walk starting from $s$ passes through at least one node in $U$ and eventually terminates at $t$. We also recall that $\vpi(s, \overline{U}, t)$ denotes the probability that such a walk from $s$ terminates at $t$ without visiting any node in $U$. Let $\tp^{(i)}(s, v)$ denote the probability that a standard (non-terminating) random walk starting at $s$ visits $v$ at its $i$-th step, that is, 
$\vpi(u, t)=\sum_{i=0}^{\infty}\alpha (1-\alpha)^i \vpi^{(i)}(u,t)$. 
Then for any $v\in \Nout(U)\setminus U$, $(1-\alpha)^{\ell}\tp^{(\ell)}(s, v)$ is an upper bound on the probability that an $\alpha$-discounted random walk starting from $s$ enters $U$ multiple times and reaches node $v$ at its $\ell$-th step without terminating. 
Moreover, since $t \notin U$, i.e., $t$ locates outside $U$, any path that passes through at least one node in $U$ must eventually exit $U$ through some vertex $v \in \Nout(U)\setminus U$ and never return to $U$ before terminating at $t$. 
Therefore, we have
\begin{align}\label{eqn:upper_sUw}
\pi(s, U, t) \le \sum_{v\in \Nout(U)\setminus U}\sum_{\ell=0}^\infty (1-\alpha)^{\ell}\tp^{(\ell)}(s, v)\pi(v, \overline{U}, t). 
\end{align} 
Recall that $\fmu(U)$ is defined as the vertex $v\in \Nout(U) \setminus U$ that maximizes $\vpi(v, \overline{U}, t)$. Therefore, $\vpi(v, \overline{U}, t)\le \vpi(\fmu(U), \overline{U}, t)$. Substituting this into~\Cref{eqn:upper_sUw} gives:
\begin{align}\label{eqn:upper_sUt}
\pi(s, U, t)
\le \sum_{v\in \Nout(U)\setminus U} \sum_{\ell=0}^\infty  (1-\alpha)^{\ell} \tp^{(\ell)}(s,v) \pi(\fmu(U), \overline{U}, t). 
\end{align} 
Since
$\sum_{v\in \Nout(U)\setminus U} \tp^{(\ell)}(s,v) \le \sum_{v\in V} \tp^{(\ell)}(s,v) = 1$, 
we further obtain: 
\begin{align*}
\pi(s, U, t) \le \sum_{\ell=0}^\infty (1-\alpha)^{\ell}\pi(\fmu(U), \overline{U}, t)=\pi(\fmu(U), \overline{U}, t)/\alpha, 
\end{align*}
which establishes~\Cref{eqn:bound_U}. The proof is completed.
\end{proof}

\begin{lemma}\label{lem:loop}
For any subset of vertices $U$ and any vertices $v, t$ in a directed graph, removing some or all in-edges of $v$ can decrease $\pi(v,t)$ and $\pi(v, \overline{U}, t)$ by at most a factor of $\alpha$. This also holds when $v$ and $t$ are identical.
\end{lemma}

\Cref{lem:loop} is used only in \Cref{sec:lowerbound_sparsegraph} and~\Cref{sec:extension_multigraphs_weighted}, where we complement our lower-bound results by incorporating $h$. We present its proof here because it is analogous to the proof of \Cref{lem:either}.

\begin{proof}[Proof of~\Cref{lem:loop}]
Let $\Ain(v)$ denote the edge set of all in-edges of $v$, that is, $\Ain(v)=\{(u,v)\mid u\in \Nin(v)\}$. Then by previous definition, $\vpi(v, \Ain(v), t)$ and $\vpi(v, \overline{\Ain(v)}, t)$ represent the probability that an $\alpha$-discounted random walk starting from $v$ terminates at $t$ with or without passing through any in-edges of $v$, respectively. We also recall from~\Cref{eqn:relation_W_barW} that 
\begin{align}\label{eqn:overline_Ain}
\vpi(v, t)=\vpi(v, \Ain(v), t)+\vpi(v, \overline{\Ain(v)}, t). 
\end{align}

Moreover, as we have defined in the proof of~\Cref{lem:either},  we use $\tp^{(i)}(u, v)$ to denote the probability that a non-terminating random walk starting at $u$ visits $v$ at step $i$. Since passing through any in-edge of $v$ will reach $v$ for sure, then $\sum_{\ell=1}^{\infty}(1-\alpha)^\ell \pi^{(i)}(v,v)$ is an upper bound on the probability that $\alpha$-discounted random walk starting at $v$ passing through $\Ain(v)$ at least once without terminating. We note that $t\notin \Ain(v)$ since $\Ain(v)$ is an edge set. So any paths that pass through $\Ain(v)$ must eventually reach $v$ and never enter $\Ain(v)$ again before terminating at $t$. Therefore, 
\begin{align*}
\vpi(v, \Ain(v), t)
&\le \sum_{i=1}^{\infty}(1-\alpha)^i \pi^{(i)}(v,v)\pi(v, \overline{\Ain(v)},t). 
\end{align*}
Combining this with~\Cref{eqn:overline_Ain} gives $\vpi(v, t)\le \sum_{i=0}^{\infty}(1-\alpha)^i \pi^{(i)}(v,v)\pi(v, \overline{\Ain(v)},t)$. Since $\pi^{(i)}(v,v)\le 1$, we further have
\begin{align*}
\vpi(v, t)
\le \sum_{i=0}^{\infty}(1-\alpha)^i \vpi(v, \overline{\Ain(v)},t)\le \vpi(v, \overline{\Ain(v)},t)/\alpha. 
\end{align*}

Furthermore, let $H$ and $H'$ denote the directed graph before and after removing some or all in-edges of $v$, respectively. We then have $\vpi_H(v, \overline{\Ain(v)},t)=\vpi_{H'}(v, \overline{\Ain(v)}, t)$, and $\vpi_{H'}(v, \overline{\Ain(v)}, t)\le \vpi_{H'}(v, t)$. This gives 
\begin{align*}
\vpi_H(v,t)\le  \vpi_{H}(v, \overline{\Ain(v)},t)/\alpha=\vpi_{H'}(v, \overline{\Ain(v)},t)/\alpha\le \vpi_{H'}(v,t)/\alpha,    
\end{align*}
thus establishing the first part of the claim in \Cref{lem:loop} concerning the decrease of $\vpi(v,t)$.

The proof of the decrease in $\vpi(v, \overline{U}, t)$ is analogous. We have
\begin{align*}
\vpi_H(v, \overline{U}, t)
&\le \sum_{i=0}^{\infty}(1-\alpha)^i \pi^{(i)}_{H}(v,v)\pi_{H}(v, \overline{U \cup \Ain(v)},t)\le \sum_{i=0}^{\infty}(1-\alpha)^i \vpi_{H}(v, \overline{U \cup \Ain(v)},t)\\
&=\vpi_H(v, \overline{U \cup \Ain(v)},t)/\alpha
=\vpi_{H'}(v, \overline{U \cup \Ain(v)},t)/\alpha \le \vpi_{H'}(v, \overline{U}, t)/\alpha, 
\end{align*}
thus completing the proof. 
\end{proof}

\subsection{Measuring Lower-Bound Complexity}
We briefly discuss here how lower-bound complexity is measured in this paper. 
We now define what it means that a vertex or vertex pair is \emph{visited} 
during an algorithm's execution.

First, a vertex pair $\{u, v\}$, whether edge or non-edge, is visited if the algorithm: 
\begin{itemize}
\item invokes one of $\outnbr$, $\innbr$, $\nonoutnbr$, or $\noninnbr$ at either $u$ or $v$ and receives the other; or 
\item invokes $\adj(u, v)$. 
\end{itemize}
A vertex $v$ is visited if 
\begin{itemize}
\item it is in a visited vertex pair $\{u, v\}$; or
\item the algorithm invokes $\indeg$ or $\outdeg$ at $v$; or 
\item the algorithm accesses the vertex $v$, by invoking $\jump$ and receiving $v$, or by directly accessing $v$. 
\end{itemize}
Regarding the last point about visiting a vertex $v$, recall that in our lower bound, we adopt the adjacency-list model~\cite[Chapter 10]{books_Goldreich17}, in which all vertices are indexed from $1$ to $n$. Thus, vertices are identified by their indices, and the algorithm can directly access any vertex by its index. In this model, $\jump$ is implemented by generating a random index.
See~\Cref{sec:preliminaries} for more details.
Since each query can visit at most $O(1)$ vertex pairs and vertices, we analyze the minimum number of vertex pairs and vertices that a good algorithm must visit to estimate $\vpi(t)$, and use this as a lower bound on the query and computational complexity of the problem. 

We now return to our lower bound proof, considering the algorithm $A$ from the assumption in~\Cref{thm:lowerbound_withn}. Note that when the random seed $R$ of $A$ is fixed, then the execution of $A$ is completely determined by
the answers it gets from querying $G$. Whether a vertex pair or a vertex is visited by $A$ is determined by the queries made during its execution. Therefore, the probability that $A$ visits a vertex pair $\{u, v\}$ or vertex $v$ in $G$ is taken over the choice of the random seed $R$ used by $A$ on $G$, and we denote these probabilities as $p_{A,G}(u, v)$ and $p_{A,G}(v)$, respectively. Accordingly, the expected number of vertex pairs and vertices in $G$ visited by $A$ is given by: 
\begin{equation*}
T_{A, G} = \sum_{\{u,v\}\in V^2} p_{A, G}(u,v) + \sum_{v \in V} p_{A,G}(v).
\end{equation*}
Therefore, $T_{A,G}$ serves as a lower bound on the expected time that algorithm $A$ spends on $G$ to estimate $\vpi(t)$.

The assumption in~\Cref{thm:lowerbound_withn} thus implies that 
\begin{align}\label{eqn:TAG_bound_min}
&T_{A,G} = O\left(\min\left\{\Teps/\log^{1/2} n,\  r/\vpi(t)\cdot \eps /\log^{3/2} n\right\}\right) \text{ for some } r \in [0,1].
\end{align}
We need to prove \Cref{thm:lowerbound_withn} for every $r\in[0,1]$, so from now on, we fix an arbitrary $r\in[0,1]$ and henceforth have that
\begin{align}\label{eqn:TAG_r}
T_{A,G} = O\left(\min\left\{\Teps/\log^{1/2} n,\  r/\vpi(t)\cdot \eps /\log^{3/2} n\right\}\right). 
\end{align}

Additionally, we have
\begin{align}\label{eqn:TAG_bound}
&T_{A,G}=\Omega(1)\cap O(\eps n/\log^{3/2}n), 
\end{align}
since no valid algorithm can achieve $T_{A,G}=o(1)$, and $\pi(t)\ge \alpha/n=\Omega(1/n)$. The inequality $\pi(t)\ge \alpha/n$ follows from~\Cref{eqn:pagerank_ppr}, along with the fact that $\vpi(t,t) \ge \alpha$, because $\vpi(t,t)$ is the probability that an $\alpha$-discounted random walk starting at $t$ terminates at $t$, which occurs with at least probability $\alpha$ if the walk stops immediately after starting from $t$.

\subsection{Basic Structures in the Graph}\label{subsec:define_notations_lowerbound}
For $\delta=1/\log^{1/5}n$, we are going to select some vertices and edges of total probability $O(\delta)$ and \emph{assume they are unvisited}. Indeed, the probability that any of them is visited by $A$ is $O(\delta)$. First, we will assume that one of the following is unvisited:

\begin{itemize}
\item \textbf{vertex} $\boldsymbol{y}$: a vertex such that $p_{A,G}(y)\le \pf$ and $\pi(y,t)\ge r$, i.e., $y \in \Veps$; 

\item \textbf{edge} $\boldsymbol{(x,y)}$: an edge in $G$ such that $p_{A,G}(x,y)\le \pf$ and $\pi(y,t)\ge r$, i.e., $y \in \Veps$; 
\end{itemize}

Since $T_{A, G} = O(\Teps / \log^{1/2} n)$, where $\Teps = \sum_{v \in \Veps} (1+\din(v))$ and $\Veps = \{v \in V \mid \vpi(v, t) \ge r\}$, at least one of $y$ and $(x,y)$ exists, and possibly both. 
To see this, if for contradiction there is no such vertex or edge, then we have
\begin{equation*}
\begin{aligned}
T_{A,G}&=\sum_{v\in V}p_{A,G}(v)+\sum_{\{u,v\}\in V^2}p_{A,G}(u,v)
> \sum_{v\in \Veps}\left(p_{A,G}(v)+\sum_{u\in \Nin(v)}p_{A,G}(u,v)/2 \right)\\
&>\pf \sum_{v\in \Veps}(1+\din(v))/2
=\pf \Teps /2. 
\end{aligned}
\end{equation*}
Since $\pf=1/\log^{1/5}n$, we then have $T_{A,G}=\omega(T_r/\log^{1/2} n)$, leading to the contradiction. If such $y$ exists, we assume $y$ is unvisited; otherwise, we assume that the edge $(x,y)$ is unvisited.

Moreover, we will show that we can identify a large subset of vertices $W \subseteq V$ in $G$, such that the total probability of visiting vertices in $W$ is $O(\pf)$. We will assume that $W$ is unvisited.

\begin{lemma}\label{lem:setW}
For $\pf = 1/\log^{1/5} n$, there exists a \textbf{vertex set} $\boldsymbol{W}\subseteq V\setminus\{y,t\}$ satisfying:  
\begin{enumerate}[label=(\alph*), font=\normalfont]
\item $\pf n /T_{A, G} \le |W|\leq 2\pf n/T_{A, G}$; 
\item $\sum_{w\in W} p_{A, G}(w)\leq 2\pf$; 
\item 
$\forall v \in V$ has fewer than $(1-\eps/2)|W|$ in-edges from $W$ and fewer than $(1-\eps/2)|W|$ out-edges to $W$.
\end{enumerate}
\end{lemma}

\begin{proof}
First, by~\Cref{eqn:TAG_bound}, we note that 
\begin{equation}\label{eqn:prob_upper}
\begin{aligned}
&\pf n/T_{A, G} = \Omega(\pf n/ (\eps n/\log^{3/2}{n}))=\omega((\log^{5/4} n)/\eps).
\end{aligned}
\end{equation}
We prove the existence of the set $W$ by showing that all conditions are satisfied with positive probability when $W$ is chosen randomly. More precisely, each vertex in $V\setminus\{y,t\}$ is selected independently for $W$ with probability $1.5\pf / T_{A, G}$. 

\paragraph{(a)} The expected size of $W$ satisfies:
\begin{align*}
\E\left[|W|\right]= 1.5(n-2)\,\pf /T_{A, G}> 1.4\,\delta n/T_{A, G}=\omega((\log^{5/4} n)/\eps).
\end{align*}
The elements are selected independently, so by the Chernoff bound, the probability that we get $|W|<\pf n/T_{A,G}$ or $|W|>2\pf n/T_{A,G}$ is $1/n^{\omega(1)}$.

\paragraph{(b)} The expected value of $\sum_{w\in W} p_{A, G}(w)$ is 
\[
 \frac{1.5\pf}{T_{A, G}}\sum_{v\in V}p_{A, G}(v)\le \frac{1.5 \pf}{T_{A, G}} T_{A, G}=1.5 \pf. 
\]
By Markov's inequality, the probability $\sum_{w\in W} p_{A, G}(w)\geq 2\pf$ is at most $3/4$.

\paragraph{(c)} 
Consider any vertex $v\in V$. 
We want to show that the probability it has 
less than $\eps \pf n/T_{A,G}$ in- or out-non-neighbors in $W$ is $1/n^{\omega(1)}$.
Then a union bound implies that each $v$ simultaneously has less than $\eps \pf n/T_{A,G}$ in- or out-non-neighbors in $W$ with probability $1/n^{\omega(1)}$.  This implies (c) when we combine with (a) stating $|W|< 2\pf n/T_{A,G}$.

The arguments for in- and out-neighbors are the same. 
The at least $\eps n$ in-non-neighbors of $v$ are picked independently for $W$, each with probability $1.5\pf / T_{A, G}$. 
The expected number of in-non-neighbors of $v$ is therefore at least
\[\eps n \cdot 1.5 \pf / T_{A, G}=\omega(\log^{5/4} n).\] 
By the Chernoff bound, the probability of falling down $\eps \pf n/ T_{A, G}$ is $1/n^{\omega(1)}$.

Adding up the error probabilities of (a), (b), and (c), we get that the total error probability is below $3/4+1/n^{\omega(1)}<1$.
\end{proof}
In the remainder of this section, we fix $W$ as the one suggested in \Cref{lem:setW} and we assume it is unvisited. The only other components assumed unvisited are $y$ or $(x,y)$, so we conclude:
\begin{lemma}\label{lem:unvisited_component}
    The total visit probability of the assumed unvisited components, that is, $y$ or $(x,y)$, and $W$, is $O(\pf)$. 
\end{lemma}

\begin{proof}
By definition, $p_{A,G}(y)\le \pf$ if $y$ is assumed unvisited; otherwise, $p_{A,G}(x,y)\le \pf$. By~\Cref{lem:setW}, $\sum_{w\in W}p_{A,G}(w)\le 2\pf$. Therefore, the total probability that $A$ visits the assumed unvisited components is at most $3\pf=O(\pf)$. 
\end{proof}

Now we define
\begin{itemize}
\item {\bf vertex $\boldsymbol y'$}: a vertex $y'\in \Nout(W)\setminus W$ that maximizes $\pi(y', \overline{W}, t)$.
\end{itemize}

Note that we do not assume that $y'$ is unvisited. Also, we note that
$y'$ may be identical to $y$. By~\Cref{lem:either} with $s=y$, we have  
\begin{align}\label{eqn:either_W}
\max\{\pi(y, \overline{W}, t),\, \pi(y', \overline{W}, t)\}=\Omega(r). 
\end{align}

\subsection{Equivalent Graphs}\label{subsec:equ_graphs}
We now consider the execution process of algorithm $A$ with a random seed $R$, denoted by $A_R$. We call a query \emph{bad} if it causes any vertices or edges that we assumed unvisited to be visited; good otherwise. 
We call a random seed $R$ {\em good} for $A$ if $A_R$ makes no bad queries. 
Recall that when we visit an edge or a non-edge, we also visit the end-vertices, so a good query does not involve any edge incident to an assumed unvisited vertex.

By \Cref{lem:unvisited_component}, 
the probability of $A$ visiting assumed unvisited components is $O(\pf)=o(1)$, since $\pf=1/\log^{1/5} n$. 
We have thus established the following lemma:

\begin{lemma}\label{lem:random_seed}
A random seed $R$ is good for $A$ with probability  $1-o(1)$.  
\end{lemma}

We now define two graphs $G$ and $G'$ to be \emph{equivalent}, 
denoted by $G \equiv G'$, if they provide the same answers to all good queries. 
Recall that when the random seed $R$ is fixed, the behavior of $A_R$ is uniquely determined by the answers $A$ obtains from its queries.  
So if a random seed $R$ for $A$ is fixed to be a good one and if $G \equiv G'$, the final estimate produced by $A_{R}$ will be exactly the same on both $G$ and $G'$. Therefore
\begin{lemma}\label{lem:seed_equi}
If $G' \equiv G$ and 
$R$ is good for $A$, then $\epi_{A_R(G')}(t)=\epi_{A_R(G)}(t)$.
\end{lemma}
Here, $\epi_{A_R(G)}(t)$ and $\epi_{A_R(G')}(t)$ denote the estimate for $\vpi(t)$ produced by running the algorithm $A$ with random seed $R$ on graph $G$ and $G'$, respectively, as defined in~\Cref{sec:preliminaries}. In combination of~\Cref{lem:random_seed}, we get

\begin{lemma}\label{lem:prob_random_seed}
If $G' \equiv G$, then 
with random seed $R$, $\Pr_R\left\{\epi_{A_R(G')}(t)=\epi_{A_R(G)}(t)\right\}=1-o(1)$, where the probability $\Pr_R$ is taken over the choice of the random seed $R$ used by the algorithm $A$. 
\end{lemma}
 
The following lemma suggests a way to construct an equivalent graph $G'$ from $G$. 

\begin{lemma}\label{lem:change_graph_general}
We have $G' \equiv G$ if $G'$ is obtained from $G$ only by changing the adjacency of an assumed unvisited vertex pairs, and by changing edges incident to assumed unvisited vertices, but preserving the in- and out-degrees of all vertices that are not assumed visited, and preserving the adjacency between all pairs not assumed unvisited. 
\end{lemma}

\begin{proof}
$G'$ may differ from $G$ only in the adjacency of assumed unvisited vertex pair, the in- and out-degrees of the assumed unvisited vertices, and the presence of edges incident to these vertices.
However, any query whose answer involve these differences, and could thus distinguish $G$ and $G'$, would necessarily causes at least one of the vertices or edges that are assumed unvisited to be visited, and hence is not good query.  This implies that the answers to all good queries will have no difference on both $G$ and $G'$, yielding $G' \equiv G$. 
\end{proof}

Recall that the assumed unvisited components are the vertex $y$ or the edge $(x,y)$, and the vertex set $W$. Therefore, \Cref{lem:change_graph_general} implies the following lemma. 

\begin{lemma}\label{lem:change_graph}
Any graph $G'$ obtained from $G$ by modifying only $(x,y)$ (if assumed unvisited), and the edges incident on $W$ and $y$ (if assumed unvisited), while preserving the in- and out-degrees of all vertices outside $W$ (except possibly the assumed unvisited $y$) ensures that $G' \equiv G$.
\end{lemma}

\begin{figure}[t]
\centering
\includegraphics[width=0.99\linewidth]{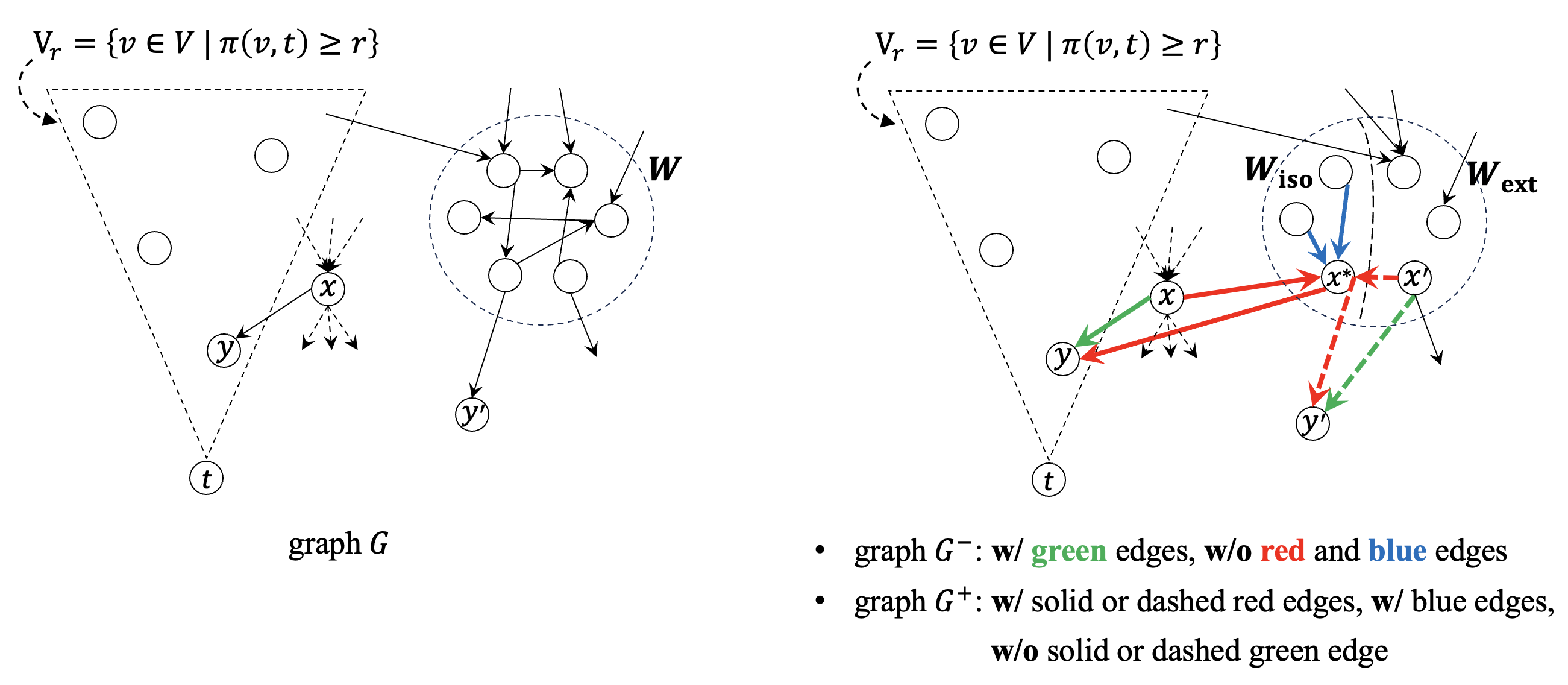}
\caption{Sketch of the constructions of the graphs $G^-$ and $G^+$} \label{fig:lower_bound_withn_constraint}
\end{figure}

\subsection{Constructing $\boldsymbol{G^+}$}\label{subsec:G-}
In this subsection, we are going to construct a graph $G^+ \equiv G$ by modifying edges in $G$ according to~\Cref{lem:change_graph}, such that $\vpi_{G^+}(t)=\omega(\vpi_G(t))$. 
Before describing the construction of $G^+$, we will first construct a graph $G^-$, which will serve as the basis for constructing $G^+$. 
A sketch of the structures of $G^-$ and $G^+$ is provided in~\Cref{fig:lower_bound_withn_constraint}. 
\begin{lemma}
\label{lem:G-}
There exists a graph $G^-\equiv G$, such that
\begin{enumerate}[label=(\alph*), font=\normalfont]
\item $G^-$ differs from $G$ only in the edges that are internal to or adjacent to $W$, while preserving the in- and out-degrees of all vertices outside $W$. 
\item there exists a subset of vertices $\Wisolated \subseteq W$ in $G^-$ containing $\lceil (\eps/2) |W|\rceil$ isolated vertices.
\end{enumerate}
\end{lemma}

\begin{proof}
Based on the structure of $G$, we construct $G^-$ by performing the following operations.

First, we remove all edges with both endpoints in $W$, i.e., edges internal to $W$. This does not change the in- or out-degrees of any vertices outside $W$. We then arbitrarily partition $W$ into two subsets, $\Wisolated$ and $\Wout = W \setminus \Wisolated$, with $|\Wisolated| = \lceil (\eps/2) |W|\rceil$ and $|\Wout| = \lfloor (1-\eps/2)|W|\rfloor$.
Next we cut all edges incident to $W$. For each vertex $v\not\in W$, by \Cref{lem:setW} (c), we know that $v$ had fewer than $(1-\eps/2)|W|$ in-neighbors in $W$, and we can now assign $v$ exactly the same number of in-neighbors in $\Wout$. This implies that the total in-degree of $v$ is unchanged. We do exactly the same for the out-neighbors and out-degree of $v$.
As a result of these changes, all nodes in $\Wisolated$ have neither in-edges nor out-edges. We denote the resulting graph as $G^-$. Note that all the above modifications are confined to vertices in $W$, and no vertices are added or deleted. Therefore, $G^-$ has the same vertex set as $G$ and differs from $G$ only in the edges internal to or adjacent to $W$, completing the proof.
\end{proof}

Using $G^-$ as the basis, we are now ready to construct $G^+$. 
We show that: 
\begin{lemma}\label{lem:G^+}
There exists a graph $G^+ \equiv G$, such that 
$\vpi_{G^+}(t)=\omega(\vpi_G(t))$. 
\end{lemma}

\begin{proof}
Recall that in~\Cref{subsec:define_notations_lowerbound}, we defined a vertex $y'\in \Nout(W)\setminus W$ that maximizes $\pi(y', \overline{W}, t)$.
By definition, $y'$ is an out-neighbor of $W$ which is not in $W$, and by~\Cref{lem:G-}, it must also be an out-neighbor of $W$ in $G^-$. We let $(x',y')$ be any edge from $W$ to $y'$ in $G^-$. Then $x'\not\in\Wisolated$.

Let $y''$ be the one of $y$ or $y'$ that maximizes $\vpi_{G^-}(y'',\overline{W},t)$. By~\Cref{lem:G-}, $G^-$ differs from $G$ only in the edges that are internal to or adjacent to $W$, while preserving the in- and out-degrees of all vertices outside $W$. This ensures $\vpi_{G^-}(y'', \overline{W}, t)=\vpi_{G}(y'', \overline{W}, t)$. 
Moreover, by~\Cref{lem:either}, we have $\vpi_{G}(y'', \overline{W}, t) = \Omega(r)$, thus implying $\vpi_{G^-}(y'', t)\ge \vpi_{G^-}(y'', \overline{W}, t)=\vpi_{G}(y'', \overline{W}, t)=\Omega(r)$.
Moreover, we arbitrarily select a vertex $x^*$ in $\Wisolated$. Our goal is to add an edge $(x^*,y'')$. 

If $y''=y$ and $y$ was assumed unvisited, then the degree of $y$ can be modified but still ensuring $G^+ \equiv G$, so we simply add $(x^*,y'')$;

Otherwise, we define $x''=x$ if $y''=y$ and $x''=x'$ if $y''=y'$. We use $x^*$ to subdivide $(x'',y'')$, that is, replacing $(x'',y'')$ by $(x'',x^*)$ and $(x^*,y'')$. We refer to the resulting graph at this stage as $G'$. By~\Cref{lem:edge_subdividing}, the subdivision decreases $\vpi_{G^-}(y'', t)$ by at most a factor $(2-\alpha)/(1-\alpha)$, so we have  $\vpi_{G'}(y'', t) = \Omega(r)$. This also implies $\vpi_{G’}(x^*, t) = \Omega(r)$ since $x^*$ in $G'$ has only one out-edge, connected directly to $y''$.

Based on $G'$, we add edges $(w, x^*)$ for every $w \in \Wisolated \setminus {x^*}$. 
We call the resulting graph $G^+$. 

Since $w \in \Wisolated \setminus {x^*}$ has no incoming edges and $t\notin W$, we must have 
$\vpi_{G^+}(x^*, t) = \vpi_{G’}(x^*, t)$. Also, $(w,x^*)$ is the only edge out of $w$, so
we conclude that 
$\vpi_{G^+}(w, t) =
(1-\alpha)\vpi_{G^+}(x^*, t) =\Omega(r)$.

Furthermore, by \Cref{lem:G-}, we had
$|\Wisolated|\geq (\eps/2)|W|$ and by \Cref{lem:setW} (a), we have $|W|\geq \delta n/T_{A,G}$, so we conclude that 
\begin{align}\label{eqn:pi_G+}
\pi_{G^+}(t)=
\sum_{v\in V}\pi_{G^+}(v,t)/n>\sum_{w\in \Wisolated}\pi_{G^+}(w,t)/n=|\Wisolated|\Omega(r)/n=\Omega(\eps r\delta /T_{A,G}).    
\end{align}
By substituting $T_{A, G} = O((r/\vpi_G(t))\eps/\log^{3/2} n)$ from~\Cref{eqn:TAG_r} and $\pf=1/\log^{1/5} n$, we obtain 
\begin{align*}
\vpi_{G^+}(t) = \Omega( \vpi_G(t) \log^{13/10} n)=\omega(\vpi_G(t)),   \end{align*}
thus finishing the proof. 
\end{proof}

\subsection{Proving~\Cref{thm:lowerbound_withn}}
\label{subsec:pi_G+}
\begin{proof}[Proof of~\Cref{thm:lowerbound_withn}]
By~\Cref{lem:G^+}, we have $G^+ \equiv G$ and
\[\vpi_{G^+}(t)=\omega(\vpi_G(t)). \]
Combining $G^+ \equiv G$ with~\Cref{lem:prob_random_seed} gives 
\begin{align*}
\Pr_R\left\{\epi_{A_R(G)}(t)=\epi_{A_R(G^+)}(t)\right\}\ge 1-o(1), 
\end{align*}
completing the proof. 
\end{proof}

\section{Upper Bounds}\label{sec:upperbound}

This section presents our upper bound for estimating $\vpi(t)$, as formally stated in~\Cref{thm:running_time_adaptive}.

\begin{theorem}\label{thm:running_time_adaptive}
For any vertex $t$ in a directed graph $G$, there exists an algorithm that estimates $\pi(t)$ within a $(1 \pm 1/\log^{1/4}{n})$ multiplicative factor w.p. at least $1 - 1/\log^{1/4}{n}$, and runs in expected time
\begin{align*}
O(T^* \log{n}) \cap \Omega(T^*). 
\end{align*}
\end{theorem}
\noindent Here, we recall from~\Cref{eqn:def_V_T} and~\Cref{eqn:T*} that 
\begin{align*}
T^*=\max_{r\in [0, 1]}\left\{\min\left\{\Teps, r/\vpi(t)\right\}\right\}, \quad \text{where} \quad  \Veps=\{v\in V \mid \vpi(v,t) \ge r\} \quad \text{and} \quad   \Teps=\sum_{v\in \Veps}\left(1+\din(v)\right). 
\end{align*}

Our algorithm is an adaptive variant of the $\BiPR$ algorithm~\cite{lofgren2016personalized}, as introduced below.

\subsection{Review of $\BiPR$ Algorithm}\label{subsec:related}
\newcommand\req[1]{(\ref{#1})}

The $\BiPR$ algorithm is a combination of $\push$~\cite{andersen2008local} and Monte Carlo simulation.

\subsubsection{The $\push$ Operation} \label{subsub:push}
\def\rpush{r_{\mathrm{push}}}
The $\push$ operation was originally proposed in~\cite{andersen2008local} for computing PageRank contributions (i.e., given a target node $t$, calculating $\vpi(v,t)$ for all $v \in V$). It propagates random-walk probability mass backward along in-edges, step by step, starting from the target vertex $t$. Each $\push$ operation maintains two variables for every vertex $v$: a ``residue'' $\r(v)\geq 0$ and a ``reserve'' $\p(v)\geq 0$. Initially, $\r(v)$ and $\p(v)$ are set to $0$ for every $v \in V$, except for $\r(t)$, which is set to $1$. The following invariant is maintained for every $s\in V$. 
\begin{equation}
\pi(s, t)=\p(s) + \sum_{v\in V}\vpi(s, v)\r(v). 
\label{eqn:invariant_ppr} 
\end{equation}
A $\push$ operation is atomic and can be applied to any vertex $v$, following the steps outlined in~\Cref{alg:push}. 
\begin{algorithm}[ht]\label{alg:push}
\DontPrintSemicolon
\caption{$\push(v)$~\cite{andersen2008local}}
\KwIn{vertex $v$}
\KwOut{updated $\r()$ and $\p()$}
$r\gets \r(v)$\;
$\p(v)\gets \p(v)+\alpha r$\; 
$\r(v)\gets 0$\; 
\For{$i$ \textup{from} $1$ \textup{to} $\indeg(v)$}{
    $u\gets \innbr(v,i)$\;
    $\r(u)\gets \r(u)+(1-\alpha)r/\outdeg(u)$\; 
}
\Return $\r()$ and $\p()$ \;
\end{algorithm}

In~\cite{andersen2008local} it is proved that invariant \req{eqn:invariant_ppr} is maintained both initially (trivial) and under any sequence of push operations. 
At any stage, we define
\begin{align}\label{eqn:def_rmax}
r_{\max}=\max_{v\in V}r(v). 
\end{align}
Since $\sum_{v\in V}\pi(s,v)=1$, invariant
\req{eqn:invariant_ppr} implies \cite[Theorem 1]{andersen2008local}:
\begin{align}\label{lem:pusherror}
\p(s) \le \vpi(s,t) \le \p(s)+\rmax. 
\end{align}
Invariant~\req{eqn:invariant_ppr} also implies an invariant on $\pi(t)$ that
\begin{align}\label{eqn:invariant}
\pi(t)=\sum_{s\in V}\pi(s,t)/n=\sum_{s\in V}\p(s)/n + \sum_{v\in V}\vpi(v)\r(v).   
\end{align}
Since $\sum_{v\in V}\pi(v)=1$, it follows
\begin{align*}
\sum_{s\in V}\p(s)/n\leq \pi(t)\leq
\sum_{s\in V}\p(s)/n + \rmax.    
\end{align*}
Thus, by performing $\push$ we can approximate $\pi(t)$ and $\pi(s,t)$ for every $s\in V$ with an additive error of $\rmax$.

One last important feature of $\push$: let $\rpush$ be the threshold such that all $\push$ are applied to vertices $v$ with $r(v)\geq \rpush$. If we continue pushing all such vertices until none remain, then we get $\rmax<\rpush$. 
Additionally, each $\push$ on $v$ adds $\alpha r(v)\geq \alpha \rpush$ to $p(v)$. By~\Cref{lem:pusherror}, $p(v) \leq \pi(v,t)$, so the number of times we can push $v$ is at most $\pi(v,t)/(\alpha\rpush)$.

\subsubsection{Combining $\push$ with Monte Carlo simulations}

The $\BiPR$ algorithm begins by performing $\push$ operations to arrive at some residuals $\r(v)$ and reserves $\p(v)$ satisfying the invariant in~\Cref{eqn:invariant}. Next, using $\r(v)$ and $\p(v)$ as fixed inputs, the algorithm independently simulates $\alpha$-discounted random walks on the graph, each starting from a uniformly sampled source vertex. For each vertex $v \in V$, let $\tpi(v)$ denote the fraction of walks that terminate at $v$. Using $\r(v)$, $\p(v)$, and $\tpi(v)$, the algorithm constructs the following Monte Carlo estimator $\epi(t)$ for $\vpi(t)$.
\begin{align}\label{eqn:bippr_estimate}
\epi(t)=\sum_{v\in V}\p(v)/n + \sum_{v\in V}\tpi(v)\r(v). 
\end{align}
For reference, we provide the pseudocode of $\BiPR$ below. 

\begin{algorithm}[H]\label{alg:bippr}
\DontPrintSemicolon
\caption{$\BiPR (t, \rmax, \nr)$~\cite{lofgren2016personalized}} 
\KwIn{target vertex $t$, threshold $\rmax$, total number of random walks $\nr$}
\KwOut{estimate for $\vpi(t)$}
$\p(),\r()\gets$ arrays with all entries $0$ \;
$\r(t)\gets 1$ \;
\For{\textup{each vertex} $v \in V$ \textup{with} $\r(v)\ge \rmax$}{
        $\push(v)$ $\quad$ \textcolor{gray}{// invoke \Cref{alg:push}}\;
}
$\nr(v) \gets 0$ for each $v\in V$ $\quad$ \textcolor{gray}{// recording the number of random walks terminating at $v$} \; 
\For{$i$ \textup{from} $1$ \textup{to} $\nr$}{
Simulate an $\alpha$-discounted random walk and use $v$ to denote the termination vertex\;
$\nr(v)\gets \nr(v)+1$\;
}
\Return $\epi(t)\gets \sum_{v\in V}\p(v)/n+\sum_{v\in V}(\nr(v)/\nr)\r(v)$ \;
\end{algorithm}

Let $\nr$ denote the total number of random walks simulated in the graph. Then $\tpi(v)$ is the average of $\nr$ independent Bernoulli random variables $\chi_v$, each taking the value $1$ with probability $\vpi(v)$. Since $\E\left[\tpi(v)\right]=\E\left[\chi_v\right]=\vpi(v)$, by invariant~\req{eqn:invariant}, the estimator $\epi(t)$ is unbiased. Moreover, as the $\chi_v$ variables are negatively correlated, the variance of $\epi(t)$ can be upper bounded as
\begin{align}\label{lem:variance_bippr}
\Var\left[\epi(t)\right]\le \sum_{v\in V}\Var\left[\tpi(v)\r(v)\right]\le \sum_{v\in V} \vpi(v)\r(v) (\rmax/\nr)\le \vpi(t) \rmax /\nr, 
\end{align}
where $r_{\max}=\max_{v\in V}r(v)$ as defined in~\Cref{eqn:def_rmax}. If $\rmax/\nr=o(\vpi(t))$
then the \emph{estimator is good} in the sense that
the standard deviation is $o(\vpi(t))$.

\subsection{Our Algorithm}
\label{subsec:adaptive_alg}
This subsection presents our upper-bound algorithm. 
Our algorithm is an adaptive variant of the $\BiPR$ algorithm. It runs in multiple rounds, where the idea for the $i$-th round is to run $\BiPR$ with a time budget proportional to $2^i$. Ideally, this budget is evenly divided between $\push$ operations and Monte Carlo simulations. We want to stop as soon as we think we have a \emph{good} estimator. 

Specifically, recall the Monte Carlo estimator from \Cref{eqn:bippr_estimate}:
\[\epi(t)=\sum_{v\in V}\p(v)/n + \sum_{v\in V}\tpi(v)\r(v). \]
If we base the estimate on $\nr$ random walks, then $\epi(t)$ is good with standard deviation $o(\pi(t))$ if $\rmax/\nr=o(\vpi(t))$. Since the true value of $\vpi(t)$ is unknown, a simple idea is to use the estimate $\epi(t)$ to test if the estimator is good. That is, we will terminate when 
\[\rmax/\nr=o(\epi(t)). \]

The desired benefit of the $\push$ operations is to decrease $\rmax$ since $\rmax$ is used to bound the variance of Monte Carlo simulations in \cref{lem:variance_bippr}. However, individual $\push$ operations may increase $\rmax$, even if we push only from the vertices $v$ with $r(v)=\rmax$. Also, we have the problem that if the next vertex $v$ that we want to push has a large in-degree, then it may be too expensive to push from such a $v$ for the current budget. 

To resolve these issues, the way we will use the pushes is that we have a budget $b$ for pushes that we increase by $2^{i-1}$ in round $i>0$. Every time we have a budget increase, we try to do as many pushes as possible while staying within budget. Since $\rmax$ can go both up and down, we will maintain a variable $\rpush$ which satisfies the following invariants. 
\begin{lemma}\label{lem:residue_upper_lower_bound}
\begin{enumerate}[label=(\alph*), font=\normalfont]
\item[(i)] All pushes done so far, have been on vertices $v$ that had $r(v)\ge \rpush$; 
\item[(ii)] At all times, all vertices $v$ have $\r(v)< 2\rpush/\alpha$. 
\end{enumerate}
\end{lemma}

We initialize as usual with $\r(v)=\p(v)=0$ for all $v$ and set $\r(t)=1$. Additionally, we initialize the push budget $\bpush=1$. 
We also maintain a variable $\rpush \in [0,1]$, which is initialized to $1$. 
Throughout the process, we ensure that all pushes are performed only on vertices $v$ such that $\r(v) \ge \rpush$. Additionally, we will have a vertex $\vpush$ which is the vertex we would like to push as soon as the budget permits. Initially, $\vpush=t$. The pseudocode for the initialization can be found in~\Cref {alg:init}.

\begin{algorithm}[H]\label{alg:init}
\DontPrintSemicolon
\caption{$\textup{\texttt{Push-Init}}(t)$}
$\r(t)=1$;\ $\p(t)=0$\; 
$\r(v)=\p(v)=0$ for each $v \neq t$\; 
$\rpush=1$\; 
$\vpush=t$\; 
$\bpush=1$\; 
\end{algorithm}

When the push-budget is increased, we start a while loop that iterates as long as $\din(\vpush)<\bpush$.
First we push $\vpush$ and subtract $\din(\vpush)+1$ from $\bpush$. Then we start looking for the next $\vpush$. We let $\vpush$ be any vertex $v$ with $r(v)\ge \rpush$. However, if there are no such vertex, we divide $\rpush$ by 2. We continue this process of dividing until we find a vertex $v$ with $\r(v)\ge \rpush$ that we use for $\vpush$ in the next iteration of our while-loop, which first tests if $\din(\vpush)<\bpush$. The pseudocode for this process is given in~\Cref{alg:increase-budget}. To ensure $\rpush$ is decreasing over time, it is important that we do not start pushing from scratch in each round, but rather that we in each round continue the pushing from previous rounds. 

\begin{algorithm}[H]\label{alg:increase-budget}
\DontPrintSemicolon
\caption{$\textup{\texttt{Increase-push-budget}}(b)$}
$\bpush=\bpush+b$ \; 
\While{$\bpush> \din(\vpush)$}{
$\push(\vpush)$ \quad \textcolor{gray}{//invoke~\Cref{alg:push}}\; 
$\bpush=\bpush-(\din(\vpush)+1)$\; 
    \While{\textup{there is no} $v$ \textup{with} $\r(v)\ge \rpush$}{
        $\rpush=\rpush/2$\;
    }
    let $\vpush$ be any vertex with $\r(\vpush)\ge \rpush$\; 
}
\end{algorithm}

After we have exhausted our push budget in round $i$, we do our Monte Carlo simulations, simulating $2^{i}$ random walks, matching the total push budget for the first $i$ rounds. After each simulation, we update the Monte Carlo estimate $\epi(t)$ according to \Cref{eqn:bippr_estimate}. The pseudocode for the Monte Carlo simulations is provided in~\Cref{alg:mc-part}. Since we are only concerned with controlling the variance, 2-independence of the random walks suffices.

More specifically, note that with high probability, each random walk requires $O(\log n)$ choices of a random outgoing edge for a total of $O(\log^2 n)$ random bits. These bits for a single random walk should be completely random and independent of each other. However, when we run $\nr$ random walks, then different random walks only need to be 2-independent of each other, so we end up using only $O(\log^2 n)$ random bits in total.

\begin{algorithm}[H]\label{alg:mc-part}
\DontPrintSemicolon
\caption{$\textup{\texttt{MonteCarlo}}_R(\nr)$}
\KwIn{the number of random walks $\nr$, a random seed $R$ for the $\nr$ random walks which should be 2-independent of each other.}
$X=\sum_{v\in V}p(v)/n$\; 
\For{$j$ \textup{from} $1$ \textup{to} $\nr$} {
Simulate an $\alpha$-discounted walk, and let $v$ be the terminal point of the walk 
\\ \textcolor{gray}{// 2-independence between any two walks suffices to bound the variance} \; 
$X=X+r(v)/\nr$\; 
}
\end{algorithm}

In round $i$, after performing all $\push$ operations, we are in fact going to construct three Monte Carlo estimates $X_\gamma$ for $\gamma\in \{1,2,3\}$. 
The first two are used to determine whether the algorithm should terminate. If we decide to terminate, the third estimate is returned as the output. 
For each $\gamma$, we use an independent random seed $R_\gamma$ specific to $\gamma$, and reuse it over all rounds $i$. 
As a result, the final estimate is independent of the estimates that decide termination. This ensures that the final estimate is unbiased. 
The full procedure of our algorithm is given in \Cref{alg:adaptive}. 
We use $2\rpush/\alpha$ instead of $\rmax$ in the stopping rule since $\rmax$ is not in the program, and we always have $\rmax=\max_{v\in V}r(v) \le 2\rpush/\alpha$ as shown in \Cref{lem:residue_upper_lower_bound}.

\begin{algorithm}[H]\label{alg:adaptive}
\DontPrintSemicolon
\caption{$\adapush(t)$}
\KwIn{target vertex $t$}
\KwOut{estimate $\epi(t)$ of $\pi(t)$}
\textbf{declare global} $\rpush$, $\bpush$, $\vpush$, $\r()$, $\p()$\;
$\textup{\texttt{Push-Init}}(t)$ \quad \textcolor{gray}{//set the push budget to $1$}\;
\For{$i=1, 2, \ldots $}{
\texttt{Increase-push-budget}($2^{i-1}$) \quad \textcolor{gray}{//the total budget over the first $i-1$ rounds is then $2^i$}\;
\For{$\gamma=1,2,3$ \textup{independently}}{
$X_{\gamma}=\texttt{MonteCarlo}_{R_\gamma}(2^{i})$  
}
$\tau=\frac{\rpush \log{n}} {\alpha 2^{i-2}}$\; 
\textbf{if} $\max\{X_1,X_2\}\geq \tau$ \textbf{then} {\bf return} $X_3$ as the estimate $\epi(t)$\;
}
\end{algorithm}

It is worth noting that the maximum number of rounds cannot exceed
\begin{align*}
i^{T}=\big\lceil \log\big({(4n/\alpha^2) \log n} \big)\big\rceil \le 1.1 \log n.    
\end{align*}
This is because after $i^T$ rounds, the time budget $2^{i^T}$ for push becomes larger than $n$, which ensures that the first $\push$ operation, the one from $t$ must have been performed. This ensures that $X_\gamma \ge \p(t)/n \ge \alpha/n$ for each $\gamma$ by~\Cref{eqn:bippr_estimate}, making the stopping rule $\max\{X_1, X_2\}\ge \tau$ trivially satisfied even when $\rpush=1$. The algorithm thus always terminates.

\subsection{Analysis}

In this subsection, we analyze the running time and approximation error of our algorithm. We begin by proving~\Cref{lem:residue_upper_lower_bound} stating that (i) $\r(v)\ge \rpush$ and (ii) $\r(v)<2\rpush /\alpha$ are invariants maintained by push.

\begin{proof}[Proof of~\Cref{lem:residue_upper_lower_bound}]
One can check that invariant (i) in~\Cref{lem:residue_upper_lower_bound} holds trivially by the design of our algorithm. The non-trivial part is to establish invariant (ii) stating for every vertex $v$, at all times, over all rounds,
$r(v)\leq 2\rpush/\alpha$. Let $\r(v)$, $\p(v)$ and $\rpush$ denote their values at the current time, and let $\rpush'=2\rpush$. For any vertex $v$,  let $\p'(v)$, $\r'(v)$ denote the values of $\p(v)$ and $\r(v)$ immediately before we halve $\rpush'$ to the current value $\rpush$. We then have $\r'(v)< \rpush'$, since we halve $\rpush$ only when the residues of all vertices are smaller than $\rpush$ at that time.  
This ensures: 
\begin{align}\label{eqn:upperbound_rlast}
\text{all vertices }v \text{ have: } \r'(v)< 2\rpush.  
\end{align}
Moreover, by the invariant~\eqref{eqn:invariant_ppr}, we have $\pi(v, t)=\p'(v) + \sum_{w\in V}\vpi(v, w)\r'(w)$.  
Combining this with \Cref{eqn:upperbound_rlast} gives that  
\begin{equation}\label{eqn:difference_psave_pi}
\text{all vertices }v \text{ have: } \pi(v, t)< \p'(v) + \sum_{w\in V}\vpi(v, w)(2 \rpush)= \p'(v)+2\rpush,  
\end{equation}
where we also apply the fact that $\sum_{w\in V}\vpi(v, w)=1$.

Note that for all $v$, $\p(v)$ can never decrease throughout the push process. This implies $\p'(v) \le \p(v)$ for all $v$ at all moments. Here, moments are only considered between the atomic $\push$ operations. 
Substituting this into \Cref{eqn:difference_psave_pi} further gives that
\begin{align}\label{eqn:psave}
\text{at all moments all }v \text{ have: } \pi(v,t) - \p(v) \le \pi(v,t) - \p'(v) < 2\rpush. 
\end{align}

Additionally, invariant~\eqref{eqn:invariant} also implies that
\begin{equation}\label{eqn:invariant_lower}
\text{at all moments all vertices }v \text{ have: } \pi(v, t)\ge \p(v) + \vpi(v, v)\r(v)\ge \p(v) + \alpha \r(v), 
\end{equation}
where we apply the fact that $\vpi(v, v)\ge \alpha$ for all $v$. 
Combining this with \Cref{eqn:psave} yields 
\begin{align}\label{eqn:differ_pi_p}
\text{at all moments all vertices }v \text{ have: } \alpha \r(v) \le \pi(v,t) - \p(v)  < 2\rpush, 
\end{align}
and thus $\alpha \r(v) < 2\rpush$ holds for all $v$ at all moments. 
This completes the proof. 
\end{proof}

\subsubsection{Setting up the round analysis}

We now set up the framework for analyzing the approximation error and expected runtime. Recall that the maximum number of rounds that are executed by our adaptive algorithm is upper bounded by $i^T\le 1.1\log{n}$. 
We then consider a variant of our algorithm that executes through the entire $i^T$ rounds, regardless of whether the stopping condition is satisfied before reaching round $i^T$. For any round $i\in [0, i^T]$, we define $\boldsymbol{\rpush^{(i)}}$, $\boldsymbol{X^{(i)}_{\gamma}}$, and $\boldsymbol{\tau^{(i)}}$ as the final value of $\rpush$, $X_\gamma$, and $\tau$ in round $i$.

With this variant, we will reach any round $i\leq i^T$ regardless of what happens in previous rounds. As a result, for each $\gamma \in \{1, 2, 3\}$, we have
\begin{align}\label{eqn:exp_variance_bound}
\E\left[X^{(i)}_{\gamma}\right] = \vpi(t), \text{ and } \Var\left[X^{(i)}_{\gamma}\right] \le (2 \rpush^{(i)}/\alpha) \vpi(t)/2^i=\frac{\vpi(t)\rpush^{(i)}}{\alpha 2^{i-1}}, 
\end{align}
where the variance bound combines \Cref{lem:variance_bippr} and the invariant (ii) in~\Cref{lem:residue_upper_lower_bound}. 

It is worth noting that even though we now pretend that we do all rounds $i\leq i^T$, we do not change the return of the algorithm. Recall that we return at the first round $i$ where $\max\{X^{(i)}_1,X^{(i)}_2\}\geq \tau^{(i)} = \frac{\rpush^{(i)}\log{n}} {\alpha 2^{i-2}}$. Let $i^R$ denote this value of $i$ when we return. The estimate returned is then $X^{(i^R)}_3$. Importantly, we note that $i^R$ only depends on the variables $X^{(\cdot)}_1$ and $X^{(\cdot)}_2$, which are independent of the variables $X^{(\cdot)}_3$. No matter how $X^{(\cdot)}_1$ and $X^{(\cdot)}_2$ are fixed, they fix $i^R$, and our estimator $X^{(i^R)}_3$ always satisfies \cref{eqn:exp_variance_bound}.

We want	to show	that our algorithm is \emph{good} in the following sense. We say the returned estimate is \emph{bad} if the returned estimate $X^{(i^R)}_3$ is not within a factor $(1 \pm 1/\log^{1/4} n)$ of $\vpi(t)$. The algorithm is \emph{good} if this happens with probability at most $1/\log^{1/4} n$ on any given input. We note that this definition of a good algorithm is stricter than the one used in the lower-bound part, where we establish a complexity lower bound for algorithms that return estimates with a constant relative error and a constant success probability. In other words, our upper-bound algorithm produces more accurate estimates while still matching the lower-bound time complexity up to logarithmic factors on most graphs.

We also want to show that it is \emph{instance-optimal} in the sense that the expected time spent up to the return in round $i^R$ is $O\left(T^* \log n\right)$ where $T^*$ is defined in~\Cref{eqn:T*}.

We will prove the goodness and efficiency of our algorithm separately in the following.

\subsubsection{Goodness}

We will first establish the goodness of our algorithm. 
\begin{lemma}\label{lem:goodness}
$\Pr\left\{|X^{(i^R)}_3 - \vpi(t)|\ge \vpi(t)/\log^{1/4} n\right\}\le 1/\log^{1/4} n$. 
\end{lemma}

\begin{proof}
Recall that 
\[\tau^{(i)}=
\frac{\rpush^{(i)}\log{n}} {\alpha 2^{i-2}}\]
is our return threshold for round $i$, that is, we return when $\max\{X^{(i)}_1,X^{(i)}_2\}\geq \tau^{(i)}$.
From \Cref{eqn:exp_variance_bound},
for each round $i$, 
we have the variance bound ${(\bar\sigma}^{(i)})^2=(\vpi(t)\rpush^{(i)})/(\alpha 2^{i-1})$. Since
$\rpush^{(i)}$ is non-increasing in $i$, the variance is decreasing in $i$. We can therefore ask how large $i$ should be to yield good estimates. 

For our analysis, we define $i^*$ as the smallest integer value such that
\begin{equation}\label{eq:i*}{(\bar\sigma}^{(i^*)})^2=\frac{\vpi(t)\rpush^{(i^*)}}{\alpha 2^{i^*-1}}\leq \frac{\vpi(t)^2}{\log n}\iff 
\vpi(t)\geq \frac{\rpush^{(i^*)}\log n}{\alpha 2^{i^*-1}}=\tau^{(i^*)}/2.
\end{equation}
We define $\sigma^*=\pi(t)/\sqrt{\log n}$, which is then an upper bound for the standard deviation for any $i\geq i^*$. Then by Chebyshev's inequality, for any $\gamma=1,2,3$ and any $i\geq i^*$,
\[
\Pr\left\{\left| X^{(i)}_\gamma(t)-\vpi(t)\right| \ge \vpi(t)/\log^{1/4}n\right\}\le \frac1{\log^{1/2} n}. \]
In particular it follows that our estimate $X^{(i^R)}_3(t)$ is good if $i^R\geq i^*$.

To finish the proof of goodness,
we need to show that the
probability that $i^R<i^*$ is $O(\log^{-1/4} n)$. This
error event happens if and only if
there is some $i<i^*$ and $\gamma \in \{1, 2\}$ such that
$X^{(i)}_\gamma>\tau^{(i)}=
\frac{\rpush^{(i)}\log{n}} {\alpha 2^{i-2}}$. Since $i^*$ was minimal
satisfying \Cref{eq:i*}, $i<i^*$ implies
$\tau^{(i)}>2\vpi(t)$, so
$X^{(i)}_\gamma>\tau^{(i)}$
implies $X^{(i)}_\gamma-\vpi(t) >\tau^{(i)}/2$ and we shall use Chebyshev to bound
the probability of this event.
To this end, we note that
our variance bound can be 
written as 
\begin{align}\label{eqn:var_pi}
{(\bar\sigma}^{(i)})^2=(\vpi(t)\rpush^{(i)})/(\alpha 2^{i-1})=\vpi(t)\tau^{(i)}/(2\log n).
\end{align}
Therefore, 
\[
\Pr\left\{\left| X^{(i)}_\gamma(t)-\vpi(t)\right| \ge \tau^{(i)}/2\right\}\le \frac{\vpi(t)\tau^{(i)}}
{2\log n}\big/(\tau^{(i)}/2)^2=
\frac{2\vpi(t)}
{\tau^{(i)}\log n}\stackrel{def}=P^{(i)}.\]
We have $\tau^{(i^*-1)}>2\vpi(t)$, so 
$P^{(i^*-1)}<1/\log n$.
Moreover, $\tau^{(i-1)}\geq 2\tau^{(i)}$, so for $i<i^*$,
$P^{(i-1)}\leq P^{(i)}/2$.
Thus we conclude that the probability that we get $i^R<i^*$ is bounded by
\[\sum_{i<i^*,\gamma=1,2} P^{(i)}\leq 4/\log n<1/\log^{1/4} n\textnormal,\]
as desired.     
\end{proof}

\subsubsection{Efficiency}
We claim that the expected run time of our algorithm is $O(2^{i^*})$. For this bound it is very important
that we terminate as soon as at least one of $X^{(i)}_1$ or $X^{(i)}_2$ gets larger than $\tau^{(i)}$. Thus, we will prove
\begin{lemma}\label{lem:efficiency_1} 
The expected time of our algorithm is $O(2^{i^*})$. 
\end{lemma}
\begin{proof}
The expected time of our algorithm is bounded by 
\[\sum_{i=1}^{i^T}2^i \Pr\{\max\{X^{(i)}_1, X^{(i)}_2\}< \tau^{(i)}\}.\] 
If we get to round $i+1$, the time
we spend is $O(2^i)$.  However, this does not happen if
$X^{(i)}_1$ or $X^{(i)}_2$ gets
larger than $\tau^{(i)}$. Note that we could also have terminated earlier if $X^{(i')}_1$ or $X^{(i')}_2$ was already larger than $\tau^{(i')}$ in previous round $i'<i+1$. But this is fine since we focus on runtime upper bound.

Let $Q^{(i)}$ be any bound on the probability that $X^{(i)}_\gamma<\tau^{(i)}$. Since $X^{(i)}_1$ and  $X^{(i)}_2$ are independent of each other, the expected cost of round $i+1$ within a constant factor is bounded by 
\begin{align}\label{eqn:runtime_i}
\prod_{\gamma=1,2}\Pr\{X^{(i)}_\gamma<\tau^{(i)}\}\,2^i\leq 
(Q^{(i)})^2 \,2^i.
\end{align}
Recall that $i^*$ is the smallest integer value such that $\tau^{(i^*)} \le 2\vpi(t)$. Moreover, 
from the definition of $\tau^{(i)}$, we have $\tau^{(i+1)}\leq \tau^{(i)}/2$.
In particular, we have that $\tau^{(i^*+2)}\leq \vpi(t)/2$.

We will focus on $i\geq i^*+2$.
Then 
$X^{(i)}_\gamma<\tau^{(i)}$ implies $X^{(i)}_\gamma<\vpi(t)/2$, or equivalently, $\vpi(t)-X^{(i)}_\gamma>\vpi(t)/2$.
Thus, by Chebyshev's inequality, 
\begin{align*}
\Pr\{X_{\gamma}^{(i)}<\tau^{(i)}\}\le \Pr\{X^{(i)}_\gamma<\vpi(t)/2\}
=\Pr\left\{\vpi(t)-X^{(i)}_\gamma>\vpi(t)/2\right\}
< \frac{{(\bar\sigma}^{(i)})^2}{(\vpi(t)/2)^2}. 
\end{align*}
Recall from~\Cref{eqn:var_pi} that ${(\bar\sigma}^{(i)})^2=\vpi(t)\tau^{(i)}/(2\log n)$. We then have: 
\begin{align*}
\Pr\{X_{\gamma}^{(i)}<\tau^{(i)}\}\le \frac{2\tau^{(i)}}{(\log n)\vpi(t)}\stackrel{def}=Q^{(i)}. 
\end{align*}
Since $\tau^{(i+1)}\le \tau^{(i)}/2$, we get $Q^{(i+1)}\le Q^{(i)}/2$, and then $(Q^{(i+1)})^2\le (Q^{(i)})^2/4$, so the expected cost for round $i+1$ in~\Cref{eqn:runtime_i} is halved when $i$ is incremented.
We conclude that the cost
for $i\geq i^*+2$ is 
dominated by the case $i=i^*+2$. 

Now $Q^{(i^*+2)}\le 1/(\log{n})$ since $\tau^{(i^*+2)} \le \vpi(t)/2$,
so we conclude that the expected cost for rounds $i+1$ 
for $i\geq i^*+2$ is $O(2^{i^*}/(\log^2 n))$. For
earlier rounds $i\leq i^*+2$, we just pay the full
cost of $O(2^i)$, so we conclude that the total cost is
$O(2^{i*})$.
\end{proof}

\subsubsection{Instance-optimality}
We will now show that our algorithm is instance optimal in the sense that $2^{i^*}=O((\log{n})T^*)$. To prove this bound, we first need to establish some lemmas.

First, we emphasize that our algorithm may push the same vertex many times, whereas our lower bound only accounts for each vertex being pushed once. However, \Cref{lem:time_push_i} shows that the number of times our algorithm can push the same vertex is bounded by $O(\log{n})$.

\begin{lemma}\label{lem:time_push_i}
A vertex is pushed at most $(3/\alpha)\log n$ times in total. 
\end{lemma}
\begin{proof}
For any vertex $v$, a $\push$ operation increases $\p(v)$ by $\alpha \r(v)$. By invariant (i) in~\Cref{lem:residue_upper_lower_bound}, we only push from $v$ when $\r(v) \ge \rpush$. Therefore, every time we push from $v$, $\p(v)$ increases by at least $\alpha \rpush$.
Additionally, we recall from \Cref{eqn:differ_pi_p} that
\begin{align}\label{eqn:upperbound_underestimate}
\text{at all moments all vertices }v \text{ have: } 0\le \vpi(v,t)-\p(v)< 2\rpush. 
\end{align}
This implies that at any moment for any $v$, the total future increment of $\p(v)$ is at most $2\rpush$. 
As a consequence, after each update to $\rpush$, the number of pushes from $v$ before $\rpush$ is halved cannot exceed $2\rpush / (\alpha \rpush) = 2/\alpha$; otherwise, $\p(v)$ would exceed $\vpi(v,t)$, violating the property that $\vpi(v,t)-\p(v)\ge 0$ as shown above.

Moreover, we recall that $\r(t)$ is initialized as $1$, so when $\rpush$ fall below $1$, at least one $\push$ on $t$ has been performed, making $\p(t) \ge \alpha$, and thus, $X_\gamma \ge \p(t)/n \ge \alpha/n$ for any $\gamma \in \{1,2,3\}$. This implies that $\rpush$ can be halved at most $\log\left(\frac{2n \log n}{\alpha^2}\right)$ times. After that, $\rpush$ would fall below $\frac{\alpha^2}{2n \log n}$, making the stopping rule $\max\{X_1,X_2\}\geq \tau = \frac{\rpush \log{n}} {\alpha 2^{i-2}}$ trivially satisfied for any $i\ge 1$. 

Recall that we do not start pushing from scratch in each round, but rather that we in each round continue the pushing from previous rounds. 
Consequently, each vertex is pushed at most $\frac{2}{\alpha}\cdot \log\left(\frac{2n \log n}{\alpha^2}\right)\le (3/\alpha) \log{n}$ total times in all rounds before the stopping condition is met and \adapush terminates.
\end{proof}

Next, in~\Cref{lem:time_push_rmax} we are going to establish a connection between $2^i$ and $T_{\rpush^{(i)}}$, where $\rpush^{(i)}$ denotes the value of $\rpush$ at the end of round $i$, as mentioned before. We also recall from~\Cref{eqn:def_V_T} that 
$\Teps=\sum_{v\in \Veps}\left(1+\din(v)\right)$, where $\Veps=\{v\in V \mid \vpi(v,t) \ge r\}$.  
For ease of presentation, we define $\Din(U)=\sum_{v\in U}(\din(u)+1)$. Then by definition, 
\begin{align}\label{eqn:def_T_D_V}
\Teps=\Din(\Veps). 
\end{align}
We will now prove
\begin{lemma}\label{lem:time_push_rmax}
For any $i \ge 1$, $T_{2\rpush^{(i)}}\leq 2^i\leq ((3/\alpha) \log n)T_{\alpha \rpush^{(i)}}$.
\end{lemma}

\begin{proof}
We define $V_{\text{push}}^{(i)}$ as the set of vertices $v$ that are pushed by the end of round $i$. 
Then $\Din(V_{\text{push}}^{(i)})$ is a lower bound on the time cost of push by round $i$, since some vertices may be pushed more than once. Recall that the push budget for the first $i$ rounds is $2^{i}$ in total. Therefore, we have $\Din(V_{\text{push}}^{(i)})\le 2^{i}$. Now note that if a vertex $v$ has $\p(v) = 0$ at the end of round $i$, then $v$ has never been pushed. By~\Cref{eqn:differ_pi_p}, we have $\vpi(v,t) <\p(v)+2\rpush^{(i)}$ for all $v$, so all the vertices $v$ with $\p(v)=0$ at the end of round $i$ satisfy $\vpi(v,t) < 2\rpush^{(i)}$, implying that the set of pushed vertices $V_{\text{push}}^{(i)}$ is a superset of $V_{2\rpush^{(i)}}=\{v \mid \vpi(v,t) \ge 2\rpush^{(i)}\}$. Thus
we have
$\Din(V_{2\rpush^{(i)}}) \le \Din(V_{\text{push}}^{(i)}) \leq 2^i$.
Combining with~\Cref{eqn:def_T_D_V} that $T_{2 \rpush^{(i)}}=\Din(V_{ 2\rpush^{(i)}})$ establishes the first bound. 

For the other bound, let $V_{\mathrm{push}+}^{(i)}$ denote the set of vertices $v$ such that the condition $\r(v) \ge \rpush^{(i)}$ was satisfied at some point by the end of round $i$. These are vertices that were ready to be pushed if we had the budget, but may not be actually pushed due to the budget constraint. We note that each $\push$ operation on a vertex $v$ costs $(\din(v) + 1)$ time, and each vertex is pushed at most $(3/\alpha) \log n$ times in round $i$, as shown in~\Cref{lem:time_push_i}. We then have 
\begin{align}\label{eqn:mid_D}
(3/\alpha) \log{n} \cdot \Din(V_{\mathrm{push}+}^{(i)}) \ge 2^i. 
\end{align} 

Furthermore, for all  $v\in V_{\mathrm{push}+}^{(i)}$, we have $\vpi(v,t)\ge \alpha \rpush^{(i)}$ since $\r(v) \ge \rpush^{(i)}$ was satisfied at some point in round $i$ by the definition of $V_{\mathrm{push}+}^{(i)}$. So the set $V_{\mathrm{push}+}^{(i)}$ is a subset of $V_{\alpha \rpush^{(i)}}=\{v \mid \vpi(v,t) \ge \alpha \rpush^{(i)}\}$. 
This implies $\Din(V_{ \alpha \rpush^{(i)}})\ge \Din(V_{\mathrm{push}+}^{(i)})$. It follows by combining with~\Cref{eqn:mid_D} that
\begin{align*}
(3/\alpha) \log{n} \cdot \Din(V_{ \alpha \rpush^{(i)}}) \ge (3/\alpha) \log{n} \cdot \Din(V_{\mathrm{push}+}^{(i)}) \ge 2^i. 
\end{align*} 
Recall from~\Cref{eqn:def_T_D_V} that $T_{\alpha \rpush^{(i)}}=\Din(V_{ \alpha \rpush^{(i)}})$. The proof is then complete. 
\end{proof}

We are now ready to prove
\begin{lemma}\label{lem:tie_i_T_new}\label{lem:tie_i_T}
$2^{i^*}=O(T^* \log{n}) \cap \Omega(T^*)$. 
\end{lemma}
\begin{proof}
Recall that we define $r^*$ as the maximizing value of $r$ in \cref{eqn:T*}. This implies $r^*/\vpi(t)=T^*$. Then $T_{r^*}\geq T^*$ while $T_{r'}\leq T^*$ for any $r'>r^*$. Note that we may have $T_{r^*}> T^*$ because $T_r$ decreases with $r$ in discrete steps. 

We first prove $2^{i^*}= O(T^* \log{n})$. We assume for contradiction that $2^{i^*-1}\ge (4/\alpha^2) (\log{n})T^{*}$, then by~\Cref{lem:time_push_rmax}, 
\begin{align*}
(3/\alpha)(\log{n})T_{\alpha \rpush^{(i^*-1)}}\ge 2^{i^*-1}\ge (4/\alpha^2)(\log{n})T^{*} \implies T_{\alpha \rpush^{(i^*-1)}}\ge T^{*} \ge T_{2r^*} \implies \alpha \rpush^{(i^*-1)} \le 2r^*. 
\end{align*}
Substituting into $r^*/\vpi(t)=T^*$ gives
\begin{align*}
T^*=\frac{r^*}{\vpi(t)}\ge \frac{(\alpha/2)\rpush^{(i^*-1)}}{\vpi(t)}. 
\end{align*}
Recall from~\Cref{eq:i*} that $\rpush^{(i^*-1)}>  \frac{(\alpha 2^{i^*-2}) \vpi(t)}{\log n}$. We then have
\begin{align*}
 T^* >\frac{\alpha^2 2^{i^*-3}}{\log{n}}=\frac{(\alpha^2/4)2^{i^*-1}}{\log{n}}, 
\end{align*}
which contracts the assumption that $2^{i^*-1}\ge (4/\alpha^2) (\log{n})T^{*}$. Therefore, we conclude that $2^{i^*-1}< (4/\alpha^2) (\log{n})T^{*}$, implying $2^{i^*}=O((\log{n})T^{*})$. 

Next, we prove $2^{i^*}= \Omega(T^*)$. 
We assume for contradiction that $2^{i^*}< T^*$. Recall from~\Cref{lem:time_push_rmax} that $T_{2\rpush^{(i^*)}}\leq 2^{i^*}$, and thus we have $T_{2\rpush^{(i^*)}}\leq 2^{i^*}< T^* \le T_{r^*}$. This implies $r^* < 2\rpush^{(i^*)}$. Substituting into $r^*/\vpi(t)=T^*$ and $\vpi(t)\geq \frac{\rpush^{(i^*)}\log n}{\alpha 2^{i^*-1}}$ from \Cref{eq:i*} gives
\begin{align*}
T^*=\frac{r^*}{\vpi(t)}< \frac{2\rpush^{(i^*)}}{\vpi(t)}\le \frac{\alpha \pi(t) 2^{i^*}}{\vpi(t) \log{n}}=\frac{\alpha 2^{i^*}}{\log{n}}. 
\end{align*}
This implies $2^{i^*}> (T^*/\alpha) \log{n}$, a contradiction. So we have $2^{i^*}\ge T^* = \Omega(T^*\log{n})$. Combining this with $2^{i^*}= O(T^*\log{n})$ finishes the proof. 
\end{proof}

\subsubsection{Proving \Cref{thm:running_time_adaptive}}

\begin{proof}[Proof of~\Cref{thm:running_time_adaptive}]
To establish the $O(T^* \log{n})$ upper bound, recall that in~\Cref{lem:goodness} we proved that the returned estimate is good. 
\Cref{lem:efficiency_1} establishes that the expected runtime of \adapush is $O(2^{i^*})$, and \Cref{lem:tie_i_T} further shows that $2^{i^*}=O(T^* \log{n})$. Combining these results gives the $O(T^* \log{n})$ runtime. 

We now establish the $\Omega(T^*)$ lower bound. Recall that the algorithm terminates at the first round $i$ such that $\max\{X^{(i)}_1,X^{(i)}_2\}\geq \tau^{(i)} = \frac{\rpush^{(i)}\log{n}} {\alpha 2^{i-2}}$. Recall from~\Cref{eqn:exp_variance_bound} that $\E\left[X^{(i)}_{1}\right] = \E\left[X^{(i)}_{2}\right] = \vpi(t)$, and 
recall from~\Cref{eq:i*} that $\vpi(t)<\tau^{(i^*-1)}/2$. Thus, for each $i<i^*$, by Markov's inequality we have 
\begin{align*}
\Pr\left\{\max\{X^{(i)}_1,X^{(i)}_2\}\geq \tau^{(i)}\right\}\le 2\Pr\left\{X^{(i)}_1\geq \tau^{(i)}\right\}\le \frac{2 \pi(t)}{\tau^{(i)}}<\frac{\tau^{(i^*-1)}}{\tau^{(i)}}=\frac{\rpush^{(i^*-1)}/\rpush^{(i)}}{2^{i^*-i-1}}\le \frac{1}{2^{i^*-i-1}}.
\end{align*}
Therefore, by a union bound, the probability that the algorithm terminates in any round $i<i^*-2$ is at most $\sum_{i<i^*-2}\frac{1}{2^{i^*-i-1}} \le 1/2$. Hence, with probability at least $1/2$, $\adapush$ reaches round $i^{*}-2$ (and thus incurs a running time of $\Omega(2^{i^*})$). By~\Cref{lem:tie_i_T}, its expected running time is $\Omega(2^{i^*})=\Omega(T^*)$. 

Combining this with the above $O(T^* \log{n})$ upper bound completes the proof. 
\end{proof}

\section{Instance-Smart Algorithm for Mostly-Degree-$\boldsymbol{n}$ Graphs}\label{sec:counterexample}

In this section, we show that the $\BiPR$-algorithm is not instance-optimal for all graphs. We demonstrate this by defining a simple class of graphs, called \emph{mostly-degree-$n$ graphs}, on which an instance-smart algorithm can perform much better, while still being good on all graphs.

The mostly-degree-$n$ graphs are defined as follows. A graph with $n$ vertices is considered a mostly-degree-$n$ graph if all but $o(n/\log n)$ of its vertices have both in-degrees and out-degrees equal to $n$. This implies that all vertices in a mostly-degree-$n$ graph have in- and out-degrees at least $n - o(n/\log n)$. As a result, performing just one $\push$ operation costs $\Omega(n)$ time. On the other hand, if no $\push$ is performed, $\Omega(n)$ random walks are required to move from a uniformly random source vertex to the target vertex $t$ with constant probability at least once. Therefore, the $\BiPR$ approach (and its variants) requires $\Omega(n)$ time for mostly-degree-$n$ graphs.

In contrast, we show that there exists an instance-smart algorithm that is good for all graphs, and for mostly-degree-$n$ graphs it runs in $O(\log n)$ time to estimate $\vpi(t)$. 

Specifically, we will first show that, for a mostly-degree-$n$ graph, the value of $\pi(t)$ for any target $t$ is $(1 \pm o(1))/n$, as detailed in~\Cref{lem:pi_degree_n}. We will then demonstrate that $O(\log n)$ time is sufficient to test whether a graph is a mostly-degree-$n$ graph, with the probability of a false positive being $o(1)$, as described in~\Cref{lem:test_degree_n}. As a result, the instance-smart algorithm first runs a test on the input graph. If the test indicates that the graph is a mostly-degree-$n$ graph, the algorithm returns $1/n$ as an approximation for $\pi(t)$. Otherwise, it invokes any good algorithm (e.g., our adaptive algorithm) to estimate $\pi(t)$.

\begin{lemma}\label{lem:pi_degree_n}  
For an arbitrary vertex $t$ in a graph with more than $n-n/\log{n}$ vertices of degree $n$, $\pi(t)=(1\pm 1/\log^{1/4}n)/n$.  
\end{lemma}

\begin{proof}
Let $\tp^{(i)}(u, v)$ be the probability that a non-terminating random walk starting at $u$ visits $v$ at step $i$. By definition, we have
\begin{align}\label{eqn:lhop}
\vpi(t)=\sum_{i=0}^{\infty}\alpha (1-\alpha)^i \sum_{u\in V}\vpi^{(i)}(u,t)/n. 
\end{align}
In the graph with more than $n-n/\log{n}$ vertices of degree $n$, 
all vertices have both in- and out-degrees in the range $[(1-1/\log{n})n, n]$, we then have
\begin{align}\label{eqn:onehop}
\vpi^{(1)}(u,v)=\frac{1}{\dout(u)}\in \left[\frac{1}{n}, \frac{1}{(1-1/\log{n})n}\right]  
\text{ for }\forall (u,v)\in E.
\end{align}
Notice that the in-degree of the target vertex $t$ is also in the range $[(1-1/\log{n})n, n]$. This implies
\begin{align}\label{eqn:onehop_sum}
(1-1/\log{n}) \le \sum_{v\in V}\vpi^{(1)}(v,t)\le \frac{1}{1-1/\log{n}}.  
\end{align}
Moreover, by definition we have
\begin{align*}
\pi^{(i)}(s,v)=\sum_{u\in V}\pi^{(j)}(s,u)\pi^{(i-j)}(u,v)
\end{align*}
for any $i\ge 0$, $j\in [0, i]$ and $s, v\in V$. 
Combining this with~\Cref{eqn:onehop} and~\Cref{eqn:onehop_sum} gives 
\begin{align*}
(1-1/\log n)^2 \le \sum_{v\in V}\pi^{(2)}(v,t)\le \left(\frac{1}{1-1/\log n}\right)^2. 
\end{align*}
By repeating the above and substituting into~\Cref{eqn:lhop}, we finally have
\begin{align*}
(\alpha/n)\sum_{i=0}^{\infty} (1-\alpha)^i (1-1/\log n)^i < \pi(t)\le (\alpha/n) \sum_{i=0}^{\infty} (1-\alpha)^i /(1-1/\log n)^i. 
\end{align*}
Therefore, the upper and lower bounds on the value of $\vpi(t)$ differ by at most a factor of
\begin{align*}
\frac{1-(1-\alpha)(1-1/\log{n})}{1-(1-\alpha)/(1-1/\log{n})}-1
=\frac{(1-\alpha)(2/\log{n}-1/\log^2{n})}{\alpha-1/\log{n}}\le \frac{4(1-\alpha)/\log{n}}{\alpha}
=o(1/\log^{1/4} n),
\end{align*}
thus completing the proof.
\end{proof}

Next, we will show that there exists an algorithm to test whether a graph is mostly-degree-$n$, such that for a graph $G$ that is not mostly-degree-$n$, the algorithm will correctly detect this with constant probability. The algorithm can only make a mistake if the input is not a mostly-degree-$n$ graph, but the test yields a (false) positive.

\begin{lemma}\label{lem:test_degree_n}
$O(\log n)$ running time is sufficient to test whether a graph has more than $n-n/\log{n}$ vertices of degree $n$, with the false-positive probability at most $1/e$.
\end{lemma}

\begin{proof}
For our test, we independently sample $\log n$ vertices with replacement from the graph. The test returns a positive result if all sampled vertices have in- and out-degrees equal to $n$, and a negative result otherwise. If the number of vertices with degrees below $n$ is more than $n/\log n$, then the probability of obtaining a positive result is at most $(1 - 1/\log n)^{\log n} \le 1/e$.     
\end{proof}

Finally, we show that $T^*=\Theta{(n)}$ on a mostly-degree-$n$ graph. 
Recall from~\Cref{eqn:def_V_T} and~\Cref{eqn:T*} that 
\begin{align*}
T^*=\max_{r\in [0, 1]}\left\{\min\left\{\Teps, r/\vpi(t)\right\}\right\}, \quad \text{where} \quad  \Veps=\{v\in V \mid \vpi(v,t) \ge r\} \quad \text{and} \quad   \Teps=\sum_{v\in \Veps}\left(1+\din(v)\right). 
\end{align*}
Consider $r=\alpha$. We have $\pi(t,t)\geq r$ since a walk from $t$ terminates instantly with probability $\alpha$. So $t\in V_r$. Since all vertices have in-degree close to $n$, we have that $T_r=\Omega(n)$. By~\Cref{lem:pi_degree_n}, we have $r/\pi(t)\sim n$, so we conclude that $T^*=\Theta(n)$. By \Cref{thm:running_time_adaptive}, $\adapush$ requires $\Omega(T^*)=\Omega(n)$ time to return a good estimate on a mostly-degree-$n$ graph. Together with~\Cref{lem:test_degree_n} shows that $\adapush$ is not instance-optimal on a mostly-degree-$n$ graph.

\section{Extensions of Lower Bounds} \label{sec:lowerbound_sparsegraph}

This section complements our lower bound results presented in~\Cref{sec:instance_lowerbound_with_n_constraints} by incorporating $h$ vertices with in- or out-degree above $(1-\eps)n$, as formulated in~\Cref{thm:lowerbound_with_sparse} below.

\begin{theorem}\label{thm:lowerbound_with_sparse}
Consider any directed graph $G$ where at most $h$ vertices have in- or out-degree above $(1-\eps)n$ for some $\eps\in[0,1]$. For any $r\in [0,1]$, suppose there exists an algorithm $A$ that estimates $\pi(t)$ in expected time $O\left(\min\left\{\Teps/((h+1)\log^{1/2} n),\  r/\vpi(t)\cdot \eps /((h+1)^2\log^{3/2} n)\right\}\right)$. 
We can then construct a graph $G^+$, such that 
\begin{align*}
\pi_{G^+}(t)=\omega(\pi_{G}(t)), \quad \text{and}\quad \Pr_R\left\{\epi_{A_R(G^+)}(t)=\epi_{A_R(G)}(t)\right\}\ge 1-o(1), 
\end{align*}
where the probability $\Pr_R$ is taken over the choice of the random seed $R$ used by the algorithm $A$. 
\end{theorem}

Let $\hset$ denote the set of vertices in $G$ whose in- or out-degrees exceed $(1-\eps)n$, so $h\ge |\hset|$. 
Recall from~\Cref{eqn:def_V_T} and~\Cref{eqn:T*} that 
\begin{align*}
T^*=\max_{r\in [0, 1]}\left\{\min\left\{\Teps, r/\vpi(t)\right\}\right\}, \quad \text{where} \quad  \Veps=\{v\in V \mid \vpi(v,t) \ge r\} \quad \text{and} \quad   \Teps=\sum_{v\in \Veps}\left(1+\din(v)\right). 
\end{align*}
Therefore, \Cref{thm:lowerbound_with_sparse} establishes a lower bound of $\Omega(T^* (\eps/(h+1)^2)/\log^{3/2}n)$.

\begin{proof}
The proof is an extension of~\Cref{thm:lowerbound_withn}. 

We need to prove the theorem for any given $r\in[0,1]$, 
so we fix an arbitrary $r\in[0,1]$ in what follows. 
\Cref{thm:lowerbound_with_sparse} then assumes that
\begin{align}\label{eqn:TAG_range}
T_{A,G} = O\left(\min\left\{\frac{\Teps}{(h+1)\log^{1/2} n},\ \frac{r \eps}{\vpi(t)(h+1)^2\log^{3/2} n} \right\}\right). 
\end{align}
We also trivially have $T_{A,G}=\Omega(1)$. Additionally, since $\pi(t)\ge \alpha/n$, it follows from~\Cref{eqn:TAG_range} that $T_{A,G} =O((\eps/(h+1)^2) n/\log^{3/2}n)$. Combining these gives 
\begin{align}\label{eqn:TAG_range_2}
T_{A,G} 
=\Omega(1)\cap O\left(\frac{n \eps}{(h+1)^2 \log^{3/2}n}\right). 
\end{align}

Analogous to our earlier construction, for $\pf=1/\log^{1/5}n$, we are going to select some vertices and edges of total probability $O(\pf)$ and {\em assume they are unvisited}. First, we assume that one of the following is unvisited:  
\begin{itemize}
\item {\bf vertex $\boldsymbol{y}$}: a vertex such that $p_{A,G}(y)\le \pf$ and $\pi(y, t)\ge r$, i.e., $y\in \Veps$; 
\item {\bf edge $\boldsymbol{(x,y)}$}: an edge in $G$ such that $p_{A,G}(x,y)=O(\pf)$, $x \in \Nin(y)\setminus \hset$, and $\pi(y, t)\ge r$, i.e., $y\in \Veps$. 
\end{itemize}
Compared with the earlier definition of $(x,y)$ in~\Cref{sec:instance_lowerbound_with_n_constraints}, we now additionally require that $x\notin \hset$.  

Since $\Veps=\{v\in V \mid \pi(v,t)\ge r\}$ and $\Teps=\sum_{v\in \Veps}(1+\din(v))$, at least one of $y$ and $(x,y)$ exists, and possibly both. If, for contradiction, there is no such vertex or edge, then we have 
\begin{align*}
T_{A,G}
&=\sum_{v\in V}p_{A,G}(v)+\sum_{\{u,v\}\in V^2} p_{A,G}(u,v)
>\sum_{v\in \Veps}\left(p_{A,G}(v)+\sum_{u\in \Nin(v) \setminus \hset}p_{A,G}(u,v)/2\right)\\
&>\pf\sum_{v\in \Veps}(1 +\max\{0,\din(v)-h\}) /2
>\delta\sum_{v\in \Veps}\frac{1+\din(v)}{2(h+1)}=\frac\delta{2(h+1)}\, \Teps.
\end{align*}
Since $\pf=1/\log^{1/5}n$, we have $T_{A,G}>\frac{\Teps}{2(h+1)\log^{1/5}n} =\omega\left(\frac{\Teps}{(h+1)\log^{1/2}n}\right)$, contradicting with~\Cref{eqn:TAG_range}. This proves the existence of at least one of $y$ and $(x,y)$. 

If such $y$ exists, we assume $y$ is unvisited; otherwise, we assume that the $(x,y)$ is unvisited.

Next, we select a large subset of vertices $W\subseteq V$ whose visiting probability is $O(\pf)$. We will assume that $W$ is unvisited.  
\begin{lemma}\label{lem:setW_sparse}
For $\pf=1/\log^{1/5}n$, there exists a \textbf{vertex set} $\boldsymbol{W} \subseteq V \setminus \{y,t\}$ satisfying
\begin{enumerate}[label=(\alph*), font=\normalfont]
\item 
$2\pf n/T_{A, G}\ge |W|\geq \pf n/T_{A, G}$
\item $\sum_{w\in W} p_{A, G}(w)\leq 2\pf$; 
\item $\forall v \in V \setminus \hset$ has fewer than $(1-\eps/2)|W|$ in-edges from $W$ and fewer than $(1-\eps/2)|W|$ out-edges to $W$. 
\end{enumerate}
\end{lemma}
Compared with the earlier definition of $W$ in~\Cref{lem:setW}, conditions (a) and (b) remain unchanged, while in (c) we restrict our attention to vertices $v \in V \setminus \hset$, since each vertex in $\hset$ has unbounded in- or out-degree and may therefore connect to every vertex in $W$. 

\begin{proof}[Proof of~\Cref{lem:setW_sparse}]
The proof is mostly the same as that of~\Cref{lem:setW}. We select each vertex in $V\setminus\{y,t\}$ independently for $W$ with probability $1.5 \pf / T_{A, G}$. We will show that the conditions (a), (b), (c) listed above are satisfied with positive probability. Notably, by \Cref{eqn:TAG_range_2}, we have 
\begin{align*}
\pf n /T_{A,G} = \Omega\left(\frac{\pf n(h+1)^2 \log^{3/2}{n}}{n \eps }\right)=\omega\left(\frac{(h+1)^2 \log^{5/4} n}{\eps}\right). 
\end{align*}

\paragraph{(a)} The expected size of $W$ satisfies:
\begin{align*}
&\E\left[|W|\right]= (n-2)\cdot 1.5\,\pf /T_{A, G}>1.4\,\delta n/T_{A, G} =\omega((\log^{5/4} n)(h+1)^2/\eps).
\end{align*}
When selecting elements independently, by the Chernoff bound, the probability of getting $|W|<\delta n/T_{A,G}$ or $|W|>2\delta n/T_{A,G}$ is $1/n^{\omega(1)}$. 

\paragraph{(b)} 
The expected value of $\sum_{w\in W} p_{A, G}(w)$ is 
\[
\frac{1.5\pf}{T_{A, G}}\sum_{v\in V}p_{A, G}(v)\le \frac{1.5 \pf}{T_{A, G}} T_{A, G}=1.5 \pf. 
\]
By Markov's inequality, the probability $\sum_{w\in W} p_{A, G}(w)\geq 2\pf$ is at most $3/4$.

\paragraph{(c)} Consider all vertices $v\in V\setminus \hset$. We want to show that the probability it has fewer than $\eps \pf n/T_{A,G}$ in- or out-non-neighbors in $W$ is $1/n^{\omega(1)}$. Then a union bound implies that all $v$ in $V \setminus \hset$ have fewer than $\eps \pf n/T_{A,G}$ in- or out-non-neighbors in $W$ with probability $1/n^{\omega(1)}$.
This implies (c) when we combine with (a) stating that $|W|\le 2\pf n/T_{A,G}$.

The arguments for in- and out-neighbors are the same. For any vertex $v\in V\setminus \hset$, the at least $\eps n$ in-non-neighbors are picked independently for $W$, each with probability $1.5 \pf / T_{A, G}$. The expected number of in-non-neighbors of $v$ in $W$ is therefore at least
\begin{align*}
1.5 \eps n \pf / T_{A, G}=\omega\left((h+1)^2 \log^{5/4}n\right).
\end{align*}
By the Chernoff bound, the probability that the number of in-non-neighbors of $v$ falls below $\eps \pf n/T_{A, G}$ is $1/n^{\omega(1)}$.

Adding up the error probabilities of (a), (b), and (c), we get that the total error probability is below $3/4+1/n^{\omega(1)}<1$.
\end{proof}

We now fix $W$ as the one suggested in~\Cref{lem:setW_sparse}, and we assume it is unvisited. We now define
\begin{itemize}
    \item {\bf vertex} $\boldsymbol{y'}$: a vertex $y' \in \Nout(W)\setminus W$ that maximizes $\pi(y', \overline{W}, t)$. 
\end{itemize}
As in our earlier construction, we do not assume that $y'$ is unvisited. Also, we note that $y'$ may be identical to $y$. By~\Cref{lem:either} with $s=y$, we again have  
\begin{align}\label{eqn:either_W-h}
\max\{\pi(y, \overline{W}, t),\, \pi(y', \overline{W}, t)\}=\Omega(r). 
\end{align}

\paragraph{Constructing the graph $\boldsymbol{G^-}$.}
We are going to construct a graph $G^- \equiv G$. 
Using $G$ as a basis, we cut all edges internal to $W$, and arbitrarily partition $W$ into $|\Wisolated| = \lceil (\eps/2) |W|\rceil$ and $|\Wout| = \lfloor (1-\eps/2)|W|\rfloor$. 
We reassign edges that originally connect vertices in $V \setminus (W \cup \hset)$ and $W$ to ensure that these edges are adjacent only to $\Wout$ after the modification. The resulting graph is referred to as $G^-$. We note that: 
\begin{enumerate}[label=(\alph*), font=\normalfont]
\item $G^-$ differs from $G$ only in the edges that are internal to or adjacent to $W$, while preserving the in- and out-degrees of all vertices outside $W$. By~\Cref{lem:change_graph}, $G^- \equiv G$. 
\item the subset of vertices $\Wisolated \subseteq W$ in $G^-$ contains $\lceil (\eps/2) |W|\rceil$ vertices, each having at most $h$ in-edges and  $h$ out-edges, with connections only to vertices in $\hset \setminus W$.
\end{enumerate}
Here, (b) holds because we reassign all edges in $G$ that connect vertices in $V \setminus (W \cup \hset)$ to be incident on $\Wout$. As a result, each $v \in \Wisolated$ can only connect to vertices in $\hset$, and $h \ge |\hset|$. 

\paragraph{Constructing the graph $\boldsymbol{G^+}$.}
We now use $G^-$ as the basis to construct a graph $G^+$, ensuring that $G^+ \equiv G^- \equiv G$. Recall that $y' \in \Nout(W)\setminus W$ is chosen to maximize $\vpi(y', \overline{W}, t)$ in $G$. When constructing $G^-$, we never remove any adjacency edges of $W$; instead, we only reassign those edges whose endpoints are not in $\hset$ to be incident on $\Wout$. Therefore, $y' \in \Nout(W)\setminus W$ still has an in-neighbor in $W$ in $G^-$. We then arbitrarily select one such in-neighbor, denoted by $x'$. Note that $x'$ may belong to $\Wisolated$ if $y' \in \hset$.

Let $y''$ be the one of $y$ and $y'$ that maximizes $\vpi_{G^-}(y'',\overline W,t)$. By~\Cref{lem:either}, we have $\vpi_{G}(y'', \overline{W}, t) = \Omega(r)$, thus implying $\vpi_{G^-}(y'', t)\ge \vpi_{G^-}(y'', \overline{W}, t)=\vpi_{G}(y'', \overline{W}, t)=\Omega(r)$.
Moreover, we arbitrarily select a vertex $x^*$ in $\Wisolated$. Our goal is to add an edge $(x^*,y'')$. Note that when $y''\in \hset$, $(x^*,y'')$ may already exist in $G^-$. Below we only discuss the case that $(x^*,y'')$ does not exist. 

If $y''=y$ and $y$ was assumed unvisited (then the in-degree of $y$ is unknown), then we simply add $(x^*,y'')$; 

Otherwise, we define $x''=x$ if $y''=y$ and $x''=x'$ if $y''=y'$.
We note that $x\notin \hset$ by the definition of $(x,y)$, and vertices such as $x^*$ in $\Wisolated$ in $G^-$ can only connect to vertices $\hset$. So $(x, x^*)$ does not exist in $G^-$. 
Moreover, all internal edges of $W$ have been removed, so $(x', x^*)$ does not exist in $G^-$ either, given that $x', x^* \in W$. As a result, $(x'', x^*)$ does not exist in $G^-$. 
So we cut $(x'',y'')$ and instead add the edges $(x'', x^*)$ and $(x^*,y'')$. This does not change the in- or out-degrees of either $x''$ or $y''$, since neither $(x'', x^*)$ nor $(x^*,y'')$ exists in $G^-$. 
We denote the resulting graph as $G'$. All vertices outside $W$ in $G'$ have the same in- and out-degrees as those in $G^-$ (except possibly the assumed unvisited $y$). Therefore, we have $G^+\equiv G^- \equiv G$. 
We also note that adding $(x^*,y'')$ (if it does not exists in $G^-$)  can only make $\vpi_{G'}(y'',\overline{W}, t)\ge \vpi_{G^-}(y'',\overline{W}, t)$. In addition, by~\Cref{lem:loop}, cutting $(x'', y'')$ still ensures that $\vpi_{G'}(y'',\overline{W}, t)\ge \alpha \vpi_{G^-}(y'', \overline{W}, t)=\Omega(r)$. Combining these, we obtain $\vpi_{G'}(y'',\overline{W}, t)=\Omega(r)$. 

Finally, we add edges $(w,x^*)$ for every $w\in \Wisolated \setminus \{x^*\}$, and we call the resulting graph $G^+$. 
Every vertex $w\in \Wisolated$ has a path of length at most two to $y''$,
and the degrees in $\Wisolated$ are at most $|\hset|+1\leq h+1$, so 
\[\vpi_{G^+}(w, t) \ge ((1-\alpha)/(h+1))^2\vpi_{G'}(y'',\overline{W},  t) =\Omega(r/(h+1)^2).\]

Since $|\Wisolated|\ge (\eps/2)|W|$ and $|W|\ge \pf n/T_{A,G}$, we conclude that 
\begin{align*}
\pi_{G^+}(t)=\sum_{v\in V}\pi_{G^+}(v,t)/n\ge\sum_{w\in \Wisolated}\pi_{G^+}(w,t)/n=|\Wisolated|\Omega(r/(h+1)^2)/n=\Omega\left( \frac{\eps r \pf}{(h+1)^2 T_{A,G}}\right). 
\end{align*}

By substituting $T_{A, G} = O\left(\frac{r \eps}{\vpi_{G}(t)(h+1)^2\log^{3/2} n}\right)$
from~\Cref{eqn:TAG_range} and $\pf=1/\log^{1/5} n$, we obtain 
\begin{align*}
\vpi_{G^+}(t) = \Omega( \vpi_G(t) \log^{13/10} n)=\omega(\vpi_G(t)),   \end{align*}
thus finishing the proof. 
\end{proof}

\section{Extensions to Multigraphs and Weighted Graphs}
\label{sec:extension_multigraphs_weighted}

This section extends our results to multigraphs and weighted graphs. 
Perhaps a bit counter-intuitively, we shall see that the bidirectional $\BiPR$ is instance-optimal on all multigraphs, but it is not instance-optimal on all simple weighted graphs. 

Note that on multigraphs and weighted graphs, we consider only standard queries, including the degree queries $\indeg$ and $\outdeg$, the neighbor queries $\innbr$ and $\outnbr$, and either the $\jump$ query (for upper bounds) or the vertex direct-access query (for lower bounds). 
We do not consider additional queries such as the adjacency query $\adj$, or non-adjacency queries $\noninnbr$ and $\nonoutnbr$.

\subsection{Multigraphs}
We first consider multigraphs, in which parallel edges between two vertices are allowed. In this setting, the degree queries $\indeg(v)$ and $\outdeg(v)$ return the in- and out-degrees of a vertex $v$, denoted by $\din(v)$ and $\dout(v)$, respectively. These correspond to the total number of incoming and outgoing edges incident to $v$, counting parallel edges with multiplicity.

Each parallel edge is stored separately in the adjacency list, without merging duplicates. Consequently, neighbor queries $\innbr(u,i)$ and $\outnbr(u,i)$ may return the same vertex $v$ for different indices $i$, due to the presence of parallel edges. We do not support queries that return the number of parallel edges between two vertices. 

We also do not support non-adjacency queries as mentioned above. 
If we supported these additional queries, then we would have the same problems as for simple graphs, for if we know that a vertex has in/out-degree $n$ and we know that the list of its non-in/out-neighbors is empty, then we know that the vertex has a single edge from/to every other vertex in the graph. 
In contrast, under the standard adjacency-list model, finding a vertex with degree $n$ on multigraphs is no longer informative: this does not imply that the vertex is adjacent to all other vertices, since the degree could instead come from $n$ parallel edges to a single vertex.

As a result, in a multigraph there may be multiple independent probabilities $p_{A,G}(u,v)$, one for each parallel edge between $u$ and $v$. Visiting any one of these parallel edges causes both $u$ and $v$ to be considered visited. Accordingly, $T_{A,G}$ is defined as the expected number of visited vertices and (multi-)edges in $G$, that is,
\begin{align*}
T_{A,G}=\sum_{(u,v)\in E} p_{A,G}(u,v)+\sum_{v\in V} p_{A,G}(v), 
\end{align*}
where each parallel edge is counted separately.

Regarding the upper bound, the runtime of the $\adapush$ (i.e., the adaptive version of $\BiPR$ descibed in~\Cref{alg:adaptive}) remains $O(T^* \log{n})$ on multigraphs, the same as that established in~\Cref{thm:running_time_adaptive}. Here, recall from~\Cref{eqn:def_V_T} and~\Cref{eqn:T*} that 
\begin{align*}
T^*=\max_{r\in [0, 1]}\left\{\min\left\{\Teps, r/\vpi(t)\right\}\right\}, \quad \text{where} \quad  \Veps=\{v\in V \mid \vpi(v,t) \ge r\} \quad \text{and} \quad   \Teps=\sum_{v\in \Veps}\left(1+\din(v)\right). 
\end{align*}

The following theorem establishes the instance-optimality of $\adapush$ on multigraphs by establishing matching lower bounds (up to polylogarithmic factors). 

\begin{theorem}\label{thm:complexity_multigraphs}
Consider any directed multigraph $G$ of order $n$. For any $r\in [0,1]$, suppose there exists an algorithm $A$ that estimates $\pi(t)$ in expected time 
\begin{align*}
T_{A,G} = O\left(\min\left\{\Teps,\  r/\vpi(t)\right\}/\log^{1/2}n \right) \text{ for some } r \in [0,1].
\end{align*}
We can then construct a graph $G^+$, such that 
\begin{align*}
\pi_{G^+}(t)=\omega(\pi_{G}(t)), \quad \text{and}\quad \Pr_R\left\{\epi_{A_R(G^+)}(t)=\epi_{A_R(G)}(t)\right\}\ge 1-o(1), 
\end{align*}
where the probability $\Pr_R$ is taken over the choice of the random seed $R$ used by the algorithm $A$. 
\end{theorem}

Since we need to prove~\Cref{thm:complexity_multigraphs} for every $r\in [0,1]$, from now on we fix an arbitrary $r\in [0,1]$. \Cref{thm:complexity_multigraphs} thus implies that
\begin{align}\label{eqn:TAG_assumption_multigraph}
T_{A,G}=O\left(\min\left\{\Teps , r/\pi(t)  \right\}/\log^{1/2}n \right). 
\end{align}
We also have
\begin{align*}
T_{A,G}=\Omega(1)\cap O(n/\log^{1/2}n), 
\end{align*}
by combining the trivial bound $\Omega(1)$ and the fact that $\pi(t)\ge \alpha/n$. 

\begin{proof}[Proof of~\Cref{thm:complexity_multigraphs}]
The proof largely follows that of~\Cref{thm:lowerbound_withn}. Let $\pf=1/\log^{1/5}n$. 
We will assume that one of the following is unvisited: 
\begin{itemize}
    \item \textbf{vertex} $\boldsymbol{y}$: a vertex such that $p_{A,G}(y)\le \pf$ and $\pi(y,t)\ge r$, i.e., $y\in \Veps$; 
    \item \textbf{edge} $\boldsymbol{(x,y)}$: an edge in $G$ such that $p_{A,G}(x,y)\le \pf$ and $\pi(y,t)\ge r$, i.e., $y\in \Veps$. 
\end{itemize} 
If $y$ exists, we assume that $y$ is unvisited. Otherwise,
we assume $(x,y)$ is unvisited. 
The proof of showing at least one of $y$ and $(x,y)$ exists is the same as the one in~\Cref{subsec:define_notations_lowerbound}. 

Moreover, we will select a subset of vertices $W \subseteq V$, and assume $W$ is unvisited. 
The definition of $W$ is slightly different from that in~\Cref{lem:setW}: 
when parallel edges are allowed, we remove the upper bound on the number in- and out-edges between each $v\in V$ and $W$, as required in (c) of~\Cref{lem:setW}. Instead, we only require that
\begin{lemma}\label{lem:setW_multigraphs}
For $\pf=1/\log^{1/5}n$, 
there exists a vertex set $W \subseteq V\setminus\{y,t\}$ satisfying:  
\begin{enumerate}[label=(\alph*), font=\normalfont]
\item $|W|\geq \pf n/T_{A, G}$; 
\item $\sum_{w\in W} p_{A, G}(w)\leq 2\pf$; 
\end{enumerate}
\end{lemma}

\Cref{lem:setW_multigraphs} follows by the same proof as parts (a) and (b) of~\Cref{lem:setW}, with the only difference being that $\eps$ is set to a constant, for example $\eps = 1/4$.

Same as before, we fix $W$ as the one suggested in~\Cref{lem:setW_multigraphs}, and assume it is unvisited. We also define
\begin{itemize}
    \item {\bf vertex} $\boldsymbol{y'}$: a vertex $y' \in \Nout(W)\setminus W$ that maximizes $\pi(y', \overline{W}, t)$. 
\end{itemize}
We do not assume that $y'$ is unvisited. By~\Cref{lem:either} with $s=y$, we again have  
\begin{align}\label{eqn:either_W-h_multi}
\max\{\pi(y, \overline{W}, t),\, \pi(y', \overline{W}, t)\}=\Omega(r). 
\end{align}

Now we are going to construct a graph $G^- \equiv G$. 
Since we now allow parallel edges, we can construct $G^-$ with $\Theta(|W|)$ isolated vertices without requiring an upper bound on the number of in- and out-edges between each $v\in V$ and $W$. 
\begin{lemma}
\label{lem:G-_multigraphs}
There exists a graph $G^-\equiv G$, such that
\begin{enumerate}[label=(\alph*), font=\normalfont]
\item $G^-$ differs from $G$ only in the edges that are internal to or adjacent to $W$, while preserving the in- and out-degrees of all vertices outside $W$. 
\item there exists a subset of vertices $\Wisolated \subseteq W$ in $G^-$ containing $\Theta(|W|)$ isolated vertices.
\end{enumerate}
\end{lemma}
\begin{proof}
We use the structure of $G$ as a basis. We first remove all internal edges within $W$, and arbitrarily partition $W$ into two subsets, $\Wisolated$ and $\Wout=W\setminus \Wisolated$, ensuring that $|\Wisolated|=\Theta(|W|)$ and $|\Wout|\ge 1$. For each $v\notin W$, we cut all of its incident edges to $W$, and add the same number of (parallel) edges to arbitrary vertices in $\Wout$. As a result, as long as $|\Wout| \ge 1$, we can always ensure that the in- and out-degree of each $v\in V \setminus W$ remain unchanged. 
\end{proof}
Using $G^-$ as the basis, we are able to construct $G^+ \equiv G$ with $\vpi_{G^+}(t)=\omega(\vpi_G(t))$. 
The construction mostly follows from~\Cref{lem:G^+}. The only difference lies in the size of $\Wisolated$. Specifically, recall that in~\Cref{eqn:pi_G+}, we have 
\begin{align*}
\pi_{G^+}(t)=|\Wisolated| \Omega(r)/n. 
\end{align*}
We now have $|\Wisolated|=\Omega(|W|)$ by~\Cref{lem:G-_multigraphs},  $|W|\ge \pf n /T_{A,G}$ by~\Cref{lem:setW_multigraphs}, $T_{A,G}=O\left(\frac{r/\pi_{G}(t)}{\log^{1/2}n}\right)$ by~\Cref{eqn:TAG_assumption_multigraph}, and $\pf=1/\log^{1/5}n$. This yields 
\begin{align*}
\pi_{G^+}(t)=\Omega( \pi_G(t) \log^{3/10}n)=\omega(\pi_G(t)). 
\end{align*}
Meanwhile, we have $G^+ \equiv G$. By~\Cref{lem:prob_random_seed}, we have $\Pr_R\left\{\epi_{A_R(G)}(t)=\epi_{A_R(G^+)}(t)\right\}\ge 1-o(1)$. 
This completes the proof. 
\end{proof}

\subsection{Weighted Graphs}
Above we show that $\adapush$ is instance-optimal for all multigraphs. One might think that the situation for simple weighted graphs should be similar in that the weight can be viewed as representing multiplicity, but it turns out that $\adapush$ is not instance-optimal on all simple weighted graphs.

\subsubsection{Counterexample on weighted graphs.}
We first present a counterexample graph on which an instance-smart algorithm outperforms $\adapush$ by more than a polylogarithmic factor in running time.

\begin{figure}[H]
\centering
\includegraphics[width=0.65\linewidth]{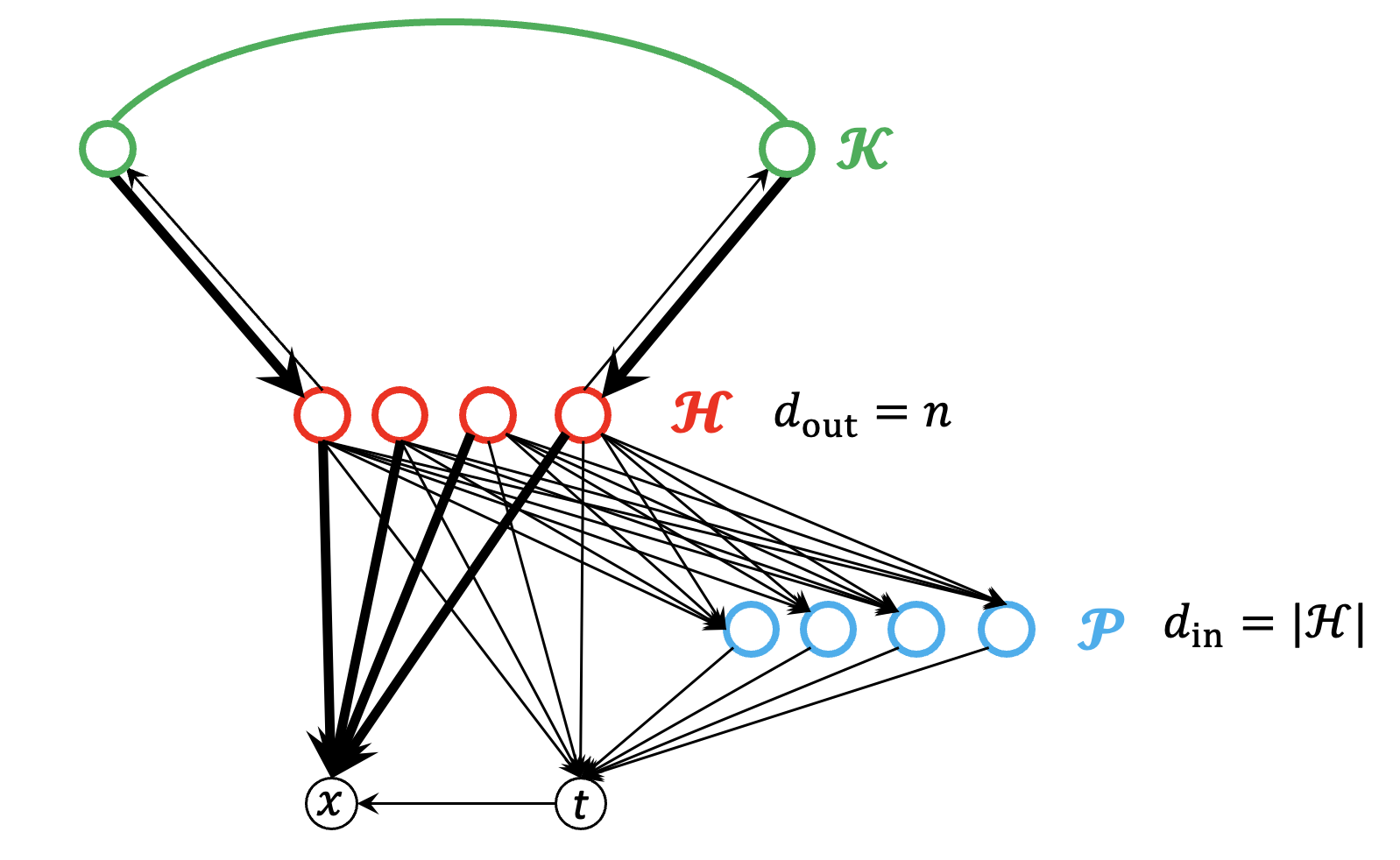}
\caption{Counterexample on weighted graphs} \label{fig:counterexample_weighted}
\end{figure}

\Cref{fig:counterexample_weighted} illustrates the counterexample. The graph consists of a target vertex $t$, a vertex $x$, a set $\mathcal{P}$ of predecessors of $t$, and a set $\mathcal{H}$ of $n$-out-degree vertices. All vertices outside $\mathcal{H}$ have a single outgoing edge to $\mathcal{H}$, so vertices in $\mathcal{H}$ accumulate large PageRank scores. We set $|\mathcal{H}| = |\mathcal{P}|$.

Each vertex in $\mathcal{P}$ has incoming edges only from vertices in $\mathcal{H}$, so its in-degree is $|\mathcal{H}|$. Each such vertex has a single outgoing edge to $t$, and $t$ has a single outgoing edge to $x$.

Each vertex in $\mathcal{H}$ has out-degree $n$. Among its outgoing edges, the edge to $x$ has weight $1 - 1/n^3$, while all other outgoing edges have negligible weight.

The vertex $x$ receives incoming edges only from $t$ and $H$, and has no outgoing edges.
All vertices outside $\mathcal{H}\cup \mathcal{P} \cup\{x,t\}$ have outgoing edges only to $\mathcal{H}$.

An instance-smart algorithm can traverse from $t$ to $x$, and then traverse every in-neighbor of $x$ (i.e., all vertices in  $\mathcal{H}$). The algorithm will observe that every vertex in $\mathcal{H}$ has out-degree $n$. 
Note that the number of vertices $n$ in the graph is known to the algorithm. 
So finding a vertex with out-degree $n$ means that such a vertex connects to all vertices in the graph. Additionally, when traversing in-neighbors of $x$, the instance-smart algorithm can notice that $x$ receives almost all incoming weight from $\mathcal{H}$. The algorithm can then perform $\push$ from $t$ to its in-neighbors. After traversing all in-neighbors of $t$, the algorithm will know that $t$ has incoming edges only from $\mathcal{H}$ and $\mathcal{P}$. Then the algorithm can traverse all vertices in $\mathcal{P}$. The algorithm will find that each vertex in $\mathcal{P}$ has in-degree $|\mathcal{H}|$, indicating that vertices in $\mathcal{P}$ receive incoming edges from every vertex in $\mathcal{H}$ and only from $\mathcal{H}$. Since the algorithm has already noticed that $x$ receives almost all incoming weight from $\mathcal{H}$, it follows that both $t$ and the vertices in $\mathcal{P}$ receive negligible weight from $\mathcal{H}$. Consequently, $\pi(t)$ is essentially proportional to $|\mathcal{P}|$.

In contrast, $\adapush$ cannot find $x$, but only performs $\push$ from $t$ to its in-neighbors. So the algorithm must verify that no edge from $\mathcal{H}$ to $\mathcal{P}$ carries a large weight. This incurs an additional $|\mathcal{H}|$ factor in the running time compared to the instance-smart algorithm. By setting $|\mathcal{H}|$ larger than any polylogarithmic factor, we rule out the possibility that $\adapush$ is instance-optimal (up to polylogarithmic factors) on such graphs.

\subsubsection{Instance-optimality}
While $\adapush$ is not instance-optimal (within logarithmic factors) for all simple weighted graphs, we do get an instance-optimality for weighted simple graphs that is a bit stronger than for simple unweighted graphs.

Recall that, for simple unweighted graphs, we proved that $\adapush$ was instance-optimal within a factor $(h+1)^2/\eps$, under the assumption that $G$ has at most $h$ vertices whose in- or out-degree exceeds $(1-\eps)n$. 
For the simple weighted graphs, we do not have any limits on the maximum in-degree, and instead assume only that $G$ has at most $h$ vertices whose {\em out-degree} exceeds $(1-\eps)n$. Under this assumption, we will prove that $\adapush$ is instance-optimal within a factor $(h+1)/\eps$.

Again, we consider only the standard adjacency-list model for weighted graphs, and do not allow additional queries such as adjacency queries or non-neighbor queries. The degree queries $\indeg(v)$ and $\outdeg(v)$ still return the in-degree and out-degree of a vertex $v$, corresponding to the number of incoming and outgoing edges incident to $v$, respectively. 

For any edge $(u,v)$, it has an associated weight $p(u,v)\le 1$, representing the probability that a walk leaving $u$ follows this edge.
For any vertex $u$, the weights $p(u)$ of all outgoing edges from $u$ sum to 1. 
The out-neighbor query $\outnbr(u,i)$ returns the $i$-th out-neighbor $v$ of $u$, along with the weight $p(u,v)$. Similarly, the in-neighbor query $\innbr(v,i)$ returns the $i$-th in-neighbor $u$ of $v$, along with the weight $p(u,v)$.

The $\adapush$ algorithm can be easily adapted to this setting. In the Monte Carlo simulations, we use the edge probabilities $p(u,v)$ when leaving $u$, instead of the uniform probability $1/\dout(u)$. In the $\push(v)$ operation, we make a corresponding modification by increasing $\r(u)$ for each in-neighbor $u$ of $v$ by $(1-\alpha)p(u,v)\r(v)$. The overall running time remains $O(T^* \log{n})$, where 
\begin{align*}
T^*=\max_{r\in [0, 1]}\left\{\min\left\{\Teps, r/\vpi(t)\right\}\right\}, \quad \text{where} \quad  \Veps=\{v\in V \mid \vpi(v,t) \ge r\} \quad \text{and} \quad   \Teps=\sum_{v\in \Veps}\left(1+\din(v)\right), 
\end{align*}
as defined in~\Cref{eqn:def_V_T} and~\Cref{eqn:T*}.

The following theorem establishes a lower bound of $\Omega(T^* ((\eps/(h+1))/\log^{3/2}n))$ for weighted graphs where at most $h$ vertices have out-degree above $(1-\eps)n$.  
Together with the $O(T^*)$ runtime of $\adapush$, this implies that $\adapush$ is instance-optimal (up to logarithmic factors) if the maximum out-degree of the graph is $O(n/\polylog{n})$. 

\begin{theorem}\label{thm:complexity_weightedgraphs}
Consider any directed weighted graph $G$ where at most $h$ vertices have out-degree above $(1-\eps)n$ for some $\eps \in [0,1]$. For any $r\in [0,1]$, suppose there exists an algorithm $A$ that estimates $\pi(t)$ in expected time 
\begin{align*}
T_{A,G} = O\left(\min\left\{\frac{\Teps}{(h+1)\log^{1/2} n},\  \frac{r \eps}{\vpi(t) \log^{3/2} n}\right\}\right).
\end{align*}
We can then construct a graph $G^+$, such that 
\begin{align*}
\pi_{G^+}(t)=\omega(\pi_{G}(t)), \quad \text{and}\quad \Pr_R\left\{\epi_{A_R(G^+)}(t)=\epi_{A_R(G)}(t)\right\}\ge 1-o(1), 
\end{align*}
where the probability $\Pr_R$ is taken over the choice of the random seed $R$ used by the algorithm $A$. 
\end{theorem}

\begin{proof}
The proof mostly follows from~\Cref{thm:lowerbound_with_sparse}. Since we need to prove the theorem for any given $r\in [0,1]$, we fix an arbitrary $r\in [0,1]$ in what follows. Then~\Cref{thm:complexity_weightedgraphs} assumes that 
\begin{align}\label{eqn:TAG_assumption_weighted}
T_{A,G} = O\left(\min\left\{\frac{\Teps}{(h+1)\log^{1/2} n},\  \frac{r \eps}{\vpi(t) \log^{3/2} n}\right\}\right).
\end{align}
We also have 
\begin{align}\label{eqn:TAG_range_weighted}
T_{A,G} = \Omega(1) \cap O\left(\frac{n \eps}{\log^{3/2}n}\right), 
\end{align}
by applying the trivial bound $\Omega(1)$ and the fact that $\pi(t)\ge \alpha/n=\Omega(1/n)$. 

We define $\hset$ as the set of vertices in $G$ whose out-degree exceeds $(1-\eps)n$, which implies that $h \ge |\hset|$. In addition, we define $\pf = 1/\log^{1/5} n$. Analogous to our earlier construction, we are going to select some vertices and edges of total probability $O(\pf)$ and assume they are unvisited. 

First, we assume that one of the following is unvisited. 
\begin{itemize}
\item {\bf vertex $\boldsymbol{y}$}: a vertex such that $p_{A,G}(y)\le \pf$ and $\pi(y, t)\ge r$, i.e., $y\in \Veps$; 
\item {\bf edge $\boldsymbol{(x,y)}$}: an edge in $G$ such that $p_{A,G}(x,y)=O(\pf)$, $x \in \Nin(y)\setminus \hset$, and $\pi(y, t)\ge r$, i.e., $y\in \Veps$. 
\end{itemize}
At least one of $y$ and $(x,y)$ exists, and the proof follows exactly from the one in~\Cref{thm:lowerbound_with_sparse}. If such $y$ exists, we assume $y$ is unvisited; otherwise, we assume that $(x,y)$ is unvisited. 

Next, we select a large subset of vertices $W\subseteq V$ whose visiting probability is $O(\pf)$. We will assume that $W$ is unvisited. 

\begin{lemma}\label{lem:setW_weightedgraphs}
For $\pf=1/\log^{1/5}n$, 
there exists a vertex set $W \subseteq V\setminus\{y,t\}$ satisfying:  
\begin{enumerate}[label=(\alph*), font=\normalfont]
\item $|W|\geq \pf n/T_{A, G}$; 
\item $\sum_{w\in W} p_{A, G}(w)\leq 2\pf$; 
\item 
No vertex in $V \setminus \hset$ has outgoing edges to all vertices in $W$. 
\end{enumerate}
\end{lemma}

Recall that in the definition of $W$ for simple unweighted graphs, we require that every vertex outside $\hset$ has at most $(1-\eps/2)|W|$ incident edges to $W$. In contrast, we now require only that, for every vertex outside $\hset$, there exists at least one vertex in $W$ that has no incoming edge from it. 

\begin{proof}
The proof is analogous to that of~\Cref{lem:setW_sparse}. We note that, by~\Cref{eqn:TAG_range_weighted} and the setting that $\delta=1/\log^{1/5}n$, we have 
\begin{align}\label{eqn:expsize_W_weighted}
\pf n / T_{A,G}=\Omega\left(\frac{\delta n \log^{3/2}n}{n\eps }\right)=\omega\left(\frac{\log^{5/4}n}{\eps}\right). 
\end{align}
We select each vertex in $V \setminus \{y,t\}$ independently for $W$ with probability $(1+1/4)\pf /T_{A,G}$. 

\paragraph{(a)} The expected size of $W$ satisfies
\begin{align*}
\E \left[|W|\right]=(n-2)(1+1/4)\pf /T_{A,G}>(1+1/5)\pf n /T_{A,G}=\omega\left((\log^{5/4}n)/\eps\right).
\end{align*}
When selecting vertices independently, by the Chernoff bound, the probability of getting $|W|<\pf n / T_{A,G}$ is $1/n^{\omega(1)}$. 

\paragraph{(b)} The expected value of $\sum_{w\in W}p_{A,G}(w)$ is
\begin{align*}
\frac{(1+1/4)\pf}{T_{A, G}}\sum_{v\in V}p_{A, G}(v)\le \frac{(1+1/4) \pf}{T_{A, G}} T_{A, G}=(1+1/4) \pf. 
\end{align*}
By Markov's inequality, the probability $\sum_{w\in W} p_{A, G}(w)\geq 2\pf$ is at most $5/8$.

\paragraph{(c)} Recall that all vertices outside $\hset$ have out-degree less than $(1-\eps)n$. Consider any vertex $u \in V \setminus \hset$. The probability that all vertices in $W$ are out-neighbors of $u$ is equal to the probability that all non-neighbors of $u$ are not included in $W$, which is at most 
\begin{align*}
(1-(5/4)\pf/T_{A,G})^{n-\dout(u)-2}, 
\end{align*}
where the $-2$ in the exponent arises because vertices $y$ and $t$ are always excluded from $W$. 
Since $\dout(u)\le (1-\eps)n$ and $\pf / T_{A,G}=\omega\left(\frac{\log^{5/4}n}{\eps n}\right)$ by~\Cref{eqn:expsize_W_weighted}, the probability is at most $1/n^{\omega(1)}$.  
Then by the union bound over the at most $n$ vertices in $V \setminus \hset$, the probability that some vertex in $V \setminus \hset$ has all vertices of $W$ as out-neighbors is at most $n/n^{\omega(1)} = 1/n^{\omega(1)}$. 

Adding up the error probabilities of (a), (b), and (c), we get that the total error probability is below $1/n^{\omega(1)}+5/8+1/n^{\omega(1)}<1$.
\end{proof}

Same as before, we fix $W$ as the one suggested in~\Cref{lem:setW_weightedgraphs}, and we assume it is unvisited. Then we define
\begin{itemize}
    \item {\bf vertex} $\boldsymbol{y'}$: a vertex $y' \in \Nout(W)\setminus W$ that maximizes $\pi(y', \overline{W}, t)$. 
\end{itemize}
As in our earlier construction, we do not assume that $y'$ is unvisited. Also, we note that $y'$ may be identical to $y$. By~\Cref{lem:either} with $s=y$, we again have  
\begin{align}\label{eqn:either_W-h-weighted}
\max\{\pi(y, \overline{W}, t),\, \pi(y', \overline{W}, t)\}=\Omega(r). 
\end{align}

In the following, we are going to construct equivalent graphs $G^-$ and $G^+$, with $G^+ \equiv G^- \equiv G$, by modifying only $(x,y)$ (if assumed unvisited), and the edges incident on $W$ and $y$ (if assumed unvisited), while preserving the in- and out-degrees of all vertices outside $W$ (except possibly the assumed unvisited $y$), as suggested in~\Cref{lem:change_graph}. Note that on weighted graphs, modifying an edge (or non-edge) includes not only changing its adjacency but also adjusting its edge weight.

\paragraph{Constructing the graph $\boldsymbol{G^-}$.}
We will use $G$ as the basis to construct a graph $G^-$ with the following properties.
\begin{lemma}
\label{lem:G-weighted}
There exists a graph $G^-\equiv G$, such that
\begin{enumerate}[label=(\alph*), font=\normalfont]
\item $G^-$ differs from $G$ only in edges that are internal to or incident on $W$, while preserving the in- and out-degrees of all vertices outside $W$.
\item every vertex in $W$ has a self-loop in $G^-$ with edge weight $1 - 1/n^2$.
\item there exists a vertex $x^* \in W$ that has no incoming edges from vertices outside $\hset \cup W$.
\end{enumerate}
\end{lemma}

\begin{proof}
Using $G$ as the basis, we cut all edges internal to $W$. 
Moreover, we arbitrarily select a vertex $x^*$ in $W$. For every vertex $u\in V\setminus (W\cup \hset)$, by~\Cref{lem:setW_weightedgraphs} (c), we know that $u$ has at most $|W|-1$ out-going edges to $W$. We cut all these outgoing edges and re-assign them to vertices in $W$, ensuring that their endpoints are not $x^*$. 

For each outgoing edge $(w,v)$ with $w \in W$, we retain all edges but reassign their weights as follows: we add a self-loop at $w$ and assign it a dominant weight $p(w,w) = 1 - 1/n^2$. The remaining weight $1/n^2$ is distributed evenly among the other outgoing edges of $w$. We denote the resulting graph by $G^-$, which satisfies all the properties in~\Cref{lem:G-weighted}. 
\end{proof}

\paragraph{Constructing the graph $\boldsymbol{G^+}$.}
We now use $G^-$ as the basis to construct a graph $G^+$ that is equivalent to $G$ and satisfies $\pi_{G^+}(t)=\omega(\pi_{G}(t))$. 

Recall that $y' \in \Nout(W)\setminus W$ is chosen to maximize $\vpi_{G}(y', \overline{W}, t)$ in $G$. Let $y''$ be whichever of $y$ and $y'$ maximizes $\vpi_{G^-}(y'',\overline{W},t)$. By~\Cref{lem:G-weighted} (a),  the graph $G^-$ coincides with $G$ on vertices outside W and on edges not incident to $W$. Therefore, $\pi_{G^-}(y'', \overline{W}, t)=\pi_{G}(y'',\overline{W},t)$. It then follows from~\Cref{lem:either} that $\vpi_{G^-}(y'', \overline{W}, t)=\vpi_{G}(y'', \overline{W}, t)=\Omega(r)$. 

As shown in~\Cref{lem:G-weighted} (c), there exists a vertex $x^* \in W$ that has no incoming edges from vertices in $V\setminus (\hset \cup W)$. Our goal is to add an edge $(x^*, y'')$ and assign it a weight 
\begin{align}\label{eqn:weight_x*y''}
p(x^*, y'')\ge 1-1/n^2, 
\end{align}
so that it dominates the total weight of the outgoing edges of $x^*$. Meanwhile, we need to preserve the in- and out-degree of $y''$ if $y''$ is not assumed unvisited. 

The edge $(x^*, y'')$ may already exist when $y''\in \hset$. In this case, since $x^* \in W$, $(x^*, y'')$ and $(x^*, x^*)$ are both unvisited, then we can modify their edge weights. We swap the edge weights between $p(x^*, x^*)$ and $p(x^*, y)$, ensuring that $p(x^*, y)=1-1/n^2$ and the total weight of the outgoing edges of $x^*$ still equals $1$. 

When $(x^*, y'')$ doesn't exist, we consider the following two cases. 
\begin{itemize}
    \item If $y''$ is assumed unvisited (then the in- and out-degree of $y$ is unknown), then we directly add the edge $(x^*, y'')$ and set its edge weight $p(x^*, y)=1-1/n^2$.
    To ensure that the total weight of the outgoing edges of $x^*$ remains $1$, we cut the self-loop $(x^*, x^*)$, noting that $p_G(x^*, x^*)=1-1/n^2$ by~\Cref{lem:G-weighted} (b).  
    \item Otherwise, in addition of cutting $(x^*,x^*)$, and adding the edge $(x^*, y'')$ with edge weight $1-1/n^2$, we also need to preserve the in-degree of $y''$. 
    \begin{itemize}
        \item If $y''=y$, we add $(x,x^*)$, set the edge weight $p(x,x^*)=p(x,y)$, and cut $(x,y)$. By this means, the in- and out-degree of both $x$ and $y$ remain unchanged, and $p(x)=p(y)=1$. 
        
        \item Otherwise (i.e., $y''=y'$), we arbitrarily select an in-neighbor $x'$ of $y'$ from $W$. Since $x'\in W$, $x'$ is unvisited, then we do not need to preserve the in- and out-degree of $x'$, but only preserve those of $y'$. So we add the weight of $(x',y')$ to the weight of $(x',x')$, and then we cut $(x',y')$. 
    \end{itemize}
\end{itemize}

We denote the resulting graph at this current stage by $G'$. Recall that $\vpi_{G^-}(y'', \overline{W}, t)=\Omega(r)$. In the transition from $G^-$ to $G'$, the edge $(x,y)$ is the only edge not incident to $W$ that may be modified. 
By~\Cref{lem:loop}, cutting $(x,y)$ ensures that $\vpi_{G'}(y'', \overline{W}, t)\ge \alpha \vpi_{G^-}(y'', \overline{W}, t)=\Omega(r)$.  
 
Finally, for every vertex $w \in W \setminus \{x^*\}$, we add the edge $(w,x^*)$, assign it the weight of the self-loop at $w$, and then we cut the self-loop. By~\Cref{lem:G-weighted}, the self-loop at each $w\in W \setminus \{x^*\}$ has weight $1-1/n^2$ in $G^-$, and this weight may be larger for $w=x'$, as described above. As a result, 
\begin{align}\label{eqn:weight_wx*}
p(w,x^*)\ge 1-1/n^2, \text{ for every vertex }w \in W \setminus \{x^*\}. 
\end{align}
We denote the resulting graph as $G^+$. 

In the transition from $G'$ to $G^+$, all modification are on edges incident on $W$, so $\vpi_{G^+}(y'', \overline{W}, t)=\vpi_{G'}(y'', \overline{W}, t)=\Omega(r)$. For each $w\in W \setminus \{x^*\}$, it at least has a two-hop path to $y''$ via $x^*$, and we have $p(w,x^*)\ge 1-1/n^2$ by~\Cref{eqn:weight_wx*} and $p(x^*,y)\ge 1-1/n^2$ by~\Cref{eqn:weight_x*y''}. Therefore, 
\begin{align*}
\pi_{G^+}(w,t)\ge (1-\alpha)^2 p(w, x^*) p(x^*, y'')\pi_{G^+}(y'',t)=\Omega(r). 
\end{align*}
It follows
\begin{align*}
\pi_{G^+}(t)=
\sum_{v\in V}\pi_{G^+}(v,t)/n>\sum_{w\in W}\pi_{G^+}(w,t)/n=|W|\Omega(r)/n.   
\end{align*}
Substituting $|W|\ge \pf n /T_{A,G}$ from~\Cref{lem:setW_weightedgraphs}, $T_{A,G} =O\left(\frac{r \eps}{\vpi_{G}(t) \log^{3/2} n}\right)$ from~\Cref{eqn:TAG_assumption_weighted}, and $\pf=1/\log^{1/5}n$ yields 
\begin{align*}
\pi_{G^+}(t)=\Omega(r\delta /T_{A,G})=\Omega((\pi_{G}(t)/\eps)\log^{13/10}n)=\omega(\pi_{G}(t)). 
\end{align*}
Meanwhile, we have $G^+ \equiv G$. By~\Cref{lem:prob_random_seed}, we have $\Pr_R\left\{\epi_{A_R(G)}(t)=\epi_{A_R(G^+)}(t)\right\}\ge 1-o(1)$. 
This completes the proof. 
\end{proof}

\bibliographystyle{alphaurl}
\bibliography{paper}

\end{document}